\documentclass[
  prx,
  twocolumn,
  superscriptaddress,
  longbibliography
]{revtex4-2}

\usepackage{amsmath,amssymb,bm,mathtools,mathrsfs}
\usepackage{amsthm}
\usepackage{physics}
\usepackage{graphicx}
\usepackage[dvipsnames]{xcolor}
\usepackage[colorlinks=true,linkcolor=blue,citecolor=blue,urlcolor=blue]{hyperref}
\usepackage{placeins}

\graphicspath{{figures/}}

\newtheorem{theorem}{Theorem}
\newtheorem{lemma}[theorem]{Lemma}

\theoremstyle{definition}

\begin{document}

\title{Measurement-Induced Phase Transitions and Logical-Space Localization in \\ Floquet Monitored Clifford Circuits}

\author{Hyunsoo Ha}
\affiliation{Department of Physics, Princeton University, Princeton, New Jersey 08544, USA}
\affiliation{Center for Theoretical Physics -- a Leinweber Institute, Massachusetts Institute of Technology, Cambridge, Massachusetts 02139, USA}

\author{David A. Huse}
\affiliation{Department of Physics, Princeton University, Princeton, New Jersey 08544, USA}

\begin{abstract}
An entanglement transition can occur without a corresponding transition in purification.
We demonstrate this separation in spatially local monitored Clifford circuits with a fixed number of layers per Floquet period, where a spatially random pattern of Clifford gates and Pauli measurements is repeated exactly in time.
Rare spatial structures persist under this repetition and break the correspondence between entanglement and purification found in conventional spacetime-random circuits. 
In one spatial dimension, the late-time entanglement produced with an initial product state remains area-law, while the system remains extensively mixed in the thermodynamic and long-time limits if the initial state is maximally mixed. In two or more dimensions, we find a transition between volume-law and area-law entanglement while the system remains mixed on both sides of this transition. An extensive amount of quantum information survives indefinitely in both phases and evolves unitarily despite the repeated measurements. We introduce the logical-support radius to characterize its spatial spreading and identify the volume-law and area-law phases with delocalized and localized logical dynamics, respectively. We further introduce the logical Krylov dimension, which counts the independent logical operators dynamically generated from an initially local operator. Its scaling resolves the transition and provides evidence for a new universality class of measurement-induced phase transitions.
\end{abstract}

\maketitle

\section{Introduction}
\label{sec:introduction}

Measurements can qualitatively alter the behavior of quantum many-body systems. Even when applied locally, measurements can reshape long-range correlations and entanglement in quantum critical states~\cite{Garratt2023,Weinstein2023,Lee2023,Yang2023,Sun2023,Murciano2023,Paviglianiti2024,Patil2024}. When measurements are repeatedly interspersed with unitary dynamics, their competition can give rise to measurement-induced phase transitions~\cite{skinner_nahum_prx2019,li_fisher_quantumzeno,LCF2019}. In the conventional spacetime-random setting, increasing the measurement rate drives a transition from a volume-law entangled phase to an area-law entangled phase. The same transition is also reflected in purification dynamics. Starting from a maximally mixed state, the system rapidly purifies in the area-law phase, whereas the purification time grows exponentially with system size in the volume-law phase~\cite{gullans2020}.

Measurement-induced transitions have been studied extensively using stabilizer states and Clifford circuits, as they enable efficient simulations at system sizes far beyond those accessible for generic quantum circuits~\cite{LCF2019,lunt2021,sierant2022,Aaronson2004}. The stabilizer formalism provides a compact description of many-body entanglement and is also foundational in quantum error correction~\cite{Fattal2004,gottesman1997}; the volume-law phase admits a quantum-error-correcting interpretation~\cite{choi2020,fan2021,li2021qec}. While their universality classes generally differ from generic circuits, Clifford circuits capture the essential competition between entanglement generation and measurement~\cite{Zabalo2022}.

Much of the literature on measurement-induced transitions has focused on spacetime-random circuits, in which the local gates and measurement locations are resampled as the dynamics proceeds. Here we instead study Floquet monitored Clifford circuits, where a single spatially-random pattern of Clifford unitaries and measurements is sampled once and then exactly repeated every period. The resulting disorder is therefore perfectly correlated in time. Related time-periodic random Clifford circuits have previously been studied in unitary settings~\cite{sunderhauf2018,kovacs2026}. Other nonunitary circuit settings with structured spacetime dependence include monitored dynamics with periodic unitary evolution~\cite{chan2019,sierant2022floquet,mochizuki2025}, spacetime-dual constructions~\cite{Ippoliti2021,Lu2021,Ippoliti2022}, crystalline Clifford circuits with spacetime-translation-invariant measurements~\cite{Sommers2023}, and periodically driven monitored free-fermion systems~\cite{Chatterjee2025}. This repeated structure is also reminiscent of Floquet quantum error-correcting codes, which are built from periodic measurement schedules~\cite{davydova2023,aasen2023,hastings2021}. Our focus, however, is not on code design, but on how Floquet repetition changes entanglement, purification, and the dynamics of the logical quantum information that survives measurement.

\begin{figure}[t]
    \centering
    \includegraphics[width=0.5\textwidth]{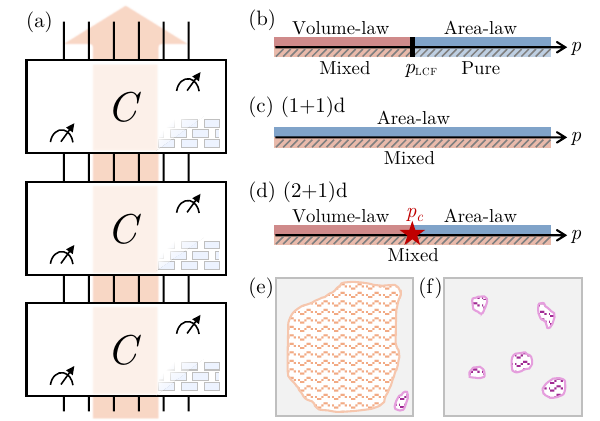}
    \caption{Schematic overview of the Floquet monitored circuit and its phase diagram.
    (a) A Floquet monitored circuit is obtained by repeating the same circuit in time. In this work, we consider circuits composed of local Clifford gates and local Pauli measurements in various dimensions. A logical subspace survives the measurements indefinitely, illustrated by the orange arrow, and can be characterized from the plateau stabilizer group as described in the main text. The Floquet evolution acts effectively unitarily within this surviving logical subspace.
    (b) Phase diagram of a generic space-time-random monitored circuit, where the gates and measurements vary randomly in both space and time. Here $p$ denotes the measurement rate. The entanglement transition from volume-law to area-law entanglement coincides with the purification transition from a mixed to a pure phase at $p_{\mathrm{LCF}}$~\cite{LCF2019,gullans2020}.
    (c) In a $(1+1)$-dimensional Floquet circuit, these two diagnostics behave differently. The system remains area-law entangled and mixed throughout the range of measurement rates.
    (d) In a $(2+1)$-dimensional Floquet circuit, the system undergoes an entanglement transition at $p_c$ from a volume-law to an area-law phase, while remaining mixed on both sides of the transition. The persistent mixed phase is associated with a nontrivial logical subspace that survives indefinitely under the repeated measurements.
    (e) Schematic of the surviving logical subspace in the volume-law phase. The logical subspace contains a delocalized component with system-wide spatial support (orange region), while spatially localized components also coexist. The volume-law phase therefore corresponds to a delocalized logical phase, characterized by a system-wide delocalized component of the surviving logical subspace.
    (f) Schematic of the surviving logical subspace in the area-law phase. In contrast to (e), the logical subspace contains no system-wide delocalized component, and all of its components remain spatially localized (purple regions). The area-law phase therefore corresponds to a localized logical phase.}
    \label{fig:Overview}
\end{figure}

We find that spatially-random Floquet repetition qualitatively separates entanglement growth from purification dynamics. In one spatial dimension, persistent spatial structures with locally high measurement density prevent entanglement from spreading throughout the system, so the late-time entanglement remains area-law for an initial product state for any nonzero measurement rate.  At the same time, other persistent regions with locally low measurement density protect a finite entropy density from measurement, so an initially maximally mixed state remains extensively mixed in the thermodynamic limit for any measurement rate less than the maximum value where every physical qubit is measured. In two (or more) dimensions, the system likewise remains mixed in the thermodynamic limit, but entanglement can propagate around local obstructions. As a result, the system exhibits a transition between volume-law and area-law entanglement without a corresponding purification transition.

To describe the quantum information that survives measurement, we use the plateau stabilizer group of Refs.~\cite{sommers2024,Fu2025}. Starting from the maximally mixed state, repeated Floquet evolution reaches a fixed stabilizer group, while the remaining logical degrees of freedom continue to evolve unitarily. Here, we explicitly construct the effective Floquet unitary acting on these degrees of freedom and represent its action on logical Pauli operators by a fixed symplectic matrix. The same logical map governs the late-time evolution from any initial stabilizer state. This formulation allows us to study directly how the surviving logical information spreads under repeated Floquet evolution.

This effective description provides a natural interpretation of the two-dimensional entanglement transition as localization within the surviving logical space. We characterize these dynamics using two complementary probes. The logical-support radius measures the size of the smallest neighborhood of the initial site that can support a physical representative of the evolved logical operator, distinguishing spatial spreading from confinement deep within the two phases. The logical Krylov dimension counts the independent logical Pauli operators generated from a local seed under repeated Floquet evolution and provides a sharp diagnostic of the transition. Together with bipartite entanglement and antipodal mutual information, these probes connect the volume-law and area-law phases to delocalized and localized logical dynamics, respectively.  In particular, initially localized logical operators all remain localized in the area-law phase, while some of them delocalize in the volume-law phase.

Floquet repetition also changes the critical behavior of the entanglement transition. The disorder is perfectly correlated in time and therefore acts as columnar disorder in spacetime. For a critical point in $D$ spatial dimensions, the Harris criterion for stability against this quenched disorder requires $\nu\geq 2/D$~\cite{harris1974,chayes1986}. This criterion does not imply that a transition must survive. In one dimension, we find no finite-$p$ entanglement transition. In two dimensions, where a transition does survive, the reported exponents $\nu<1$ of the spacetime-random MIPT violate the $D=2$ Harris bound~\cite{turkeshi2020,lunt2021,sierant2022}. The two-dimensional Floquet transition is therefore expected to belong to a universality class distinct from that of the spacetime-random transition. The strong system-size growth of the AMI peak, in contrast to the $O(1)$ critical AMI of spacetime-random circuits, provides numerical evidence for this distinction. The system-size dependence of the time required for entanglement growth to saturate also shows no evidence of infinite-randomness criticality over the accessible sizes~\cite{zabalo2023,Shkolnik2023}. Together, these results provide evidence for a new disorder-controlled critical point.

The remainder of the paper is organized as follows:  In Sec.~\ref{sec:plateau_logical}, we introduce the plateau stabilizer group and derive the effective Floquet unitary dynamics of the surviving logical degrees of freedom. In Sec.~\ref{sec:one_dimension}, we show how persistent spatial structures separate entanglement and purification in one dimension. In Sec.~\ref{sec:two_dimensional_transition}, we turn to two dimensions and establish the volume-law to area-law entanglement transition using bipartite entanglement and antipodal mutual information. In Sec.~\ref{sec:deep_phases}, we introduce the logical support radius and use it to characterize the spatial spreading and confinement of logical operators. In Sec.~\ref{sec:logical_krylov}, we introduce the logical Krylov dimension and use it to probe the localization transition in the surviving logical space. In Sec.~\ref{sec:discussion}, we discuss the implications of these results and how they may change under weak departures from the Clifford limits. Technical proofs and additional numerical results are presented in the Appendices.

\section{Plateau Stabilizer Group and Effective Floquet Dynamics}
\label{sec:plateau_logical}

We first review the plateau stabilizer group of Refs.~\cite{sommers2024} and then characterize the effective Floquet dynamics of the surviving logical degrees of freedom. The plateau group specifies the surviving logical space, while the Floquet circuit induces a fixed, invertible Clifford transformation on its logical Pauli operators. This provides a dynamical description of how the surviving quantum information evolves despite the repeated measurements.

Throughout this work, we take Pauli operators modulo overall phases in $\{\pm1,\pm i\}$ and keep only the corresponding stabilizer groups. In particular, we discard the $\pm1$ eigenvalue information carried by stabilizer signs. This representation is sufficient for the entanglement entropies and full system entropy considered below, which depend only on the stabilizer group structure. With this convention, the stabilizer group update is independent of the measurement outcomes. A fixed monitored Clifford circuit therefore defines a deterministic map on these groups. We denote the circuit corresponding to one Floquet period by $\mathcal{C}_F$.

\subsection{Plateau Stabilizer Group}
\label{sec:plateau_stabilizer_group}

The maximally mixed state has the trivial stabilizer group
$\mathrm{ISG}_0=\langle\hat{\mathbb{I}}\rangle$. We denote the instantaneous stabilizer group after $t$ Floquet periods by $\mathrm{ISG}_t$, so that
\begin{align}
    \mathrm{ISG}_{t+1}
    =
    \mathcal{C}_F[\mathrm{ISG}_t]~.
    \label{eq:ISG_update}
\end{align}

Both Clifford unitaries and Pauli measurements preserve inclusion of stabilizer groups in this representation, as shown in Appendix~\ref{app:inclusion_property}. Thus,
\begin{align}
    \mathcal{S}_1\subseteq\mathcal{S}_2
    \quad\Longrightarrow\quad
    \mathcal{C}_F[\mathcal{S}_1]
    \subseteq
    \mathcal{C}_F[\mathcal{S}_2]~.
    \label{eq:inclusion_property}
\end{align}
Since $\mathrm{ISG}_0$ is the trivial stabilizer group, it is contained in every stabilizer group and in particular
$\mathrm{ISG}_0\subseteq\mathrm{ISG}_1$. Repeated application of
Eq.~\eqref{eq:inclusion_property} therefore gives
\begin{align}
    \mathrm{ISG}_0
    \subseteq
    \mathrm{ISG}_1
    \subseteq
    \mathrm{ISG}_2
    \subseteq\cdots.
    \label{eq:ISG_sequence}
\end{align}

Each strict inclusion increases the stabilizer rank by at least one, while the rank of a stabilizer group on $N$ qubits cannot exceed $N$. There is therefore an earliest time $t^*\leq N$ at which the sequence stops growing. We define
\begin{align}
    \mathrm{ISG}_\infty
    \equiv
    \mathrm{ISG}_{t^*}
    =
    \mathrm{ISG}_{t^*+1}~,
    \label{eq:plateau_stabilizer_group}
\end{align}
and call $\mathrm{ISG}_\infty$ the \emph{plateau stabilizer group}. Because the same Floquet map is applied every period, equality at $t^*$ implies
$\mathrm{ISG}_t=\mathrm{ISG}_\infty$ for all $t\geq t^*$.

For a mixed stabilizer state on $N$ qubits, the full system entropy is
\begin{align}
    S(t)
    =
    N-\operatorname{rank}(\mathrm{ISG}_t)~.
    \label{eq:purification_entropy_rank}
\end{align}
The nested sequence in Eq.~\eqref{eq:ISG_sequence} therefore makes $S(t)$ nonincreasing. Moreover, because two nested stabilizer groups of equal rank are identical, $t^*$ is precisely the first time at which the full system entropy reaches its final plateau. We therefore call $t^*$ the \emph{plateau time}.

\subsection{Logical Floquet Dynamics}
\label{sec:logical_floquet_dynamics}

Since
$\mathcal{C}_F[\mathrm{ISG}_\infty]=\mathrm{ISG}_\infty$,
the same plateau stabilizer group is present after every Floquet period once the plateau is reached. The corresponding logical Pauli space is therefore preserved, although the circuit can act nontrivially within it. The number of surviving logical qubits is
\begin{align}
    k
    =
    S(t^*)
    =
    N-\operatorname{rank}(\mathrm{ISG}_\infty)~.
    \label{eq:number_logical_qubits}
\end{align}
The corresponding logical operator space is
\begin{align}
    \mathcal{L}_\infty
    =
    \mathrm{ISG}_\infty^\perp/
    \mathrm{ISG}_\infty~,
    \label{eq:logical_pauli_space}
\end{align}
where $\mathrm{ISG}_\infty^\perp$ denotes the Pauli operators commuting with every element of $\mathrm{ISG}_\infty$. Physical Pauli operators $P$ that differ by a plateau stabilizer represent the same logical operator. Thus, for $P\in\mathrm{ISG}_\infty^\perp$, we denote its logical equivalence class by $[P]$, with $[P]=[PS]$ for any $S\in\mathrm{ISG}_\infty$. In particular, every element of $\mathrm{ISG}_\infty$ represents the logical identity. Commutation between logical Pauli operators defines the usual symplectic form over the binary field $\mathbb{F}_2$. The resulting logical Pauli space is a symplectic vector space of dimension $2k$.

The plateau stabilizer group specifies which logical degrees of freedom survive, but not how they evolve. To define their Floquet dynamics, consider a nontrivial logical Pauli operator $[P]$ and the stabilizer group $\langle\mathrm{ISG}_\infty,P\rangle$, generated by $\mathrm{ISG}_\infty$ together with $P$. Since $P$ is nontrivial in the logical space, this group has rank one larger than $\mathrm{ISG}_\infty$. One Floquet period maps this group to another stabilizer group containing the same plateau stabilizer group,
\begin{align}
    \mathcal{C}_F
    \left[
        \langle\mathrm{ISG}_\infty, P\rangle
    \right]
    =
    \langle\mathrm{ISG}_\infty, P'\rangle~.
    \label{eq:logical_floquet_unitary}
\end{align}
Here $P'$ is again a nontrivial logical Pauli operator.
Replacing $P$ by $PS$ with $S\in\mathrm{ISG}_\infty$ leaves the enlarged stabilizer group unchanged, so the image $[P']$ depends only on the logical class $[P]$. The map $[P]\mapsto[P']$ therefore defines the Floquet evolution of the logical Pauli operators.

Appendix~\ref{app:logical_map} shows that this map is an invertible Clifford transformation of $\mathcal{L}_\infty$. A single Clifford representative $\mathcal{U}_F$ can be chosen such that
$[P']=[\mathcal{U}_F P\mathcal{U}_F^\dagger]$ for every $[P]\in\mathcal{L}_\infty$
and
$\mathcal{U}_F[\mathrm{ISG}_\infty]=\mathrm{ISG}_\infty$.
More generally, the same representative describes one Floquet period for any stabilizer group $\mathcal{S}$ containing the plateau stabilizer group,
\begin{align}
    \mathrm{ISG}_\infty\subseteq\mathcal{S}
    \quad\Longrightarrow\quad
    \mathcal{C}_F[\mathcal{S}]
    =
    \mathcal{U}_F[\mathcal{S}]~,
    \label{eq:CF_equals_UF}
\end{align}
where the equality is understood for stabilizer groups with phases discarded, as throughout this work.

At the level of full stabilizer \textit{states}, the Clifford representative generally depends on the measurement outcomes because they determine the stabilizer signs~\cite{aasen2023}. Discarding these signs leaves a unique induced action on the logical Pauli classes. In a chosen logical Pauli basis, this action is represented by a symplectic matrix $M_F\in\mathrm{Sp}(2k,\mathbb{F}_2)$, which we call the logical Floquet matrix.

The matrix $M_F$ is obtained by applying Eq.~\eqref{eq:logical_floquet_unitary} to a complete basis of $2k$ logical Pauli operators. The numerical construction is described in Appendix~\ref{app:numerical_map}. Once obtained, the $2k\times2k$ matrix can be iterated directly to generate the logical Floquet dynamics without repeatedly evolving the full monitored circuit. For the parameters considered here, $k$ is substantially smaller than the total number of physical qubits, which makes this reduction computationally useful and allows access to larger system sizes. In later sections, we explicitly construct and iterate $M_F$ in the numerical calculations to access larger system sizes.

\subsection{Long Time Dynamics from Arbitrary Initial States}
\label{sec:long_time_dynamics}

The logical Floquet dynamics applies to any initial stabilizer group $\mathcal{S}_0$, not only the maximally mixed state. Since $\langle\hat{\mathbb{I}}\rangle\subseteq\mathcal{S}_0$, repeated application of the inclusion property in Eq.~\eqref{eq:inclusion_property}, together with $\mathrm{ISG}_t=\mathrm{ISG}_\infty$ for $t\geq t^*$, gives
\begin{align}
    \mathrm{ISG}_\infty
    =
    \mathcal{C}_F^t[\langle\hat{\mathbb{I}}\rangle]
    \subseteq
    \mathcal{C}_F^t[\mathcal{S}_0]~,
    \qquad
    t\geq t^*.
    \label{eq:arbitrary_state_contains_PSG}
\end{align}
Thus, after the plateau time, the evolved stabilizer group always contains $\mathrm{ISG}_\infty$. Equation~\eqref{eq:CF_equals_UF} therefore applies to every subsequent Floquet period. The stabilizer group evolves under $\mathcal{U}_F$, with the induced dynamics on the surviving logical degrees of freedom represented by $M_F$.

In particular, consider evolution from a pure stabilizer state. Each measurement trajectory remains pure, so its stabilizer group has rank $N$ at all times. For $t\geq t^*$, a generating set can be chosen such that $N-k$ generators span the fixed plateau stabilizer group and the remaining $k$ generators specify a pure stabilizer state of the $k$ surviving logical qubits. The plateau stabilizer group remains fixed, while the remaining $k$ logical stabilizers evolve under repeated application of $M_F$.

\section{Rare-Region Separation of Entanglement and Purification in One Dimension}

\label{sec:one_dimension}
\begin{figure*}
    \centering
    \includegraphics[width=\textwidth]{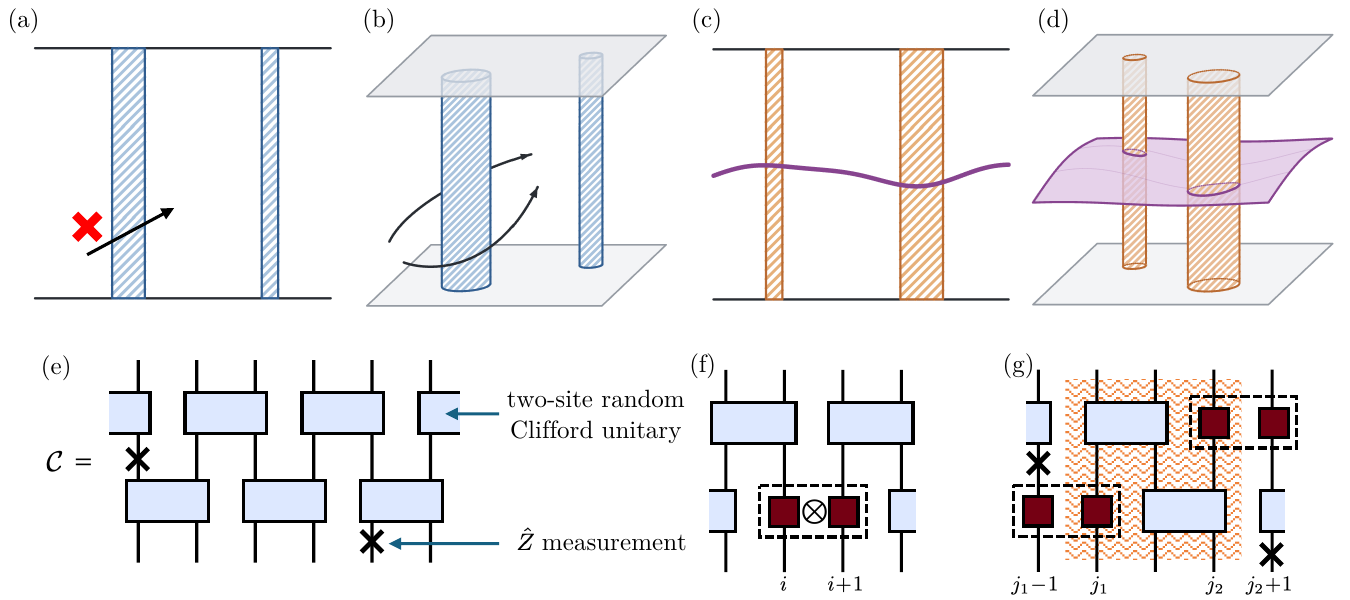}
    \caption{Persistent blockers and storage regions.
    (a) In one dimension, a finite blocker (blue) persists under Floquet repetition and prevents entanglement from being produced across it.
    (b) In two (or more) dimensions, entanglement can be produced along paths that avoid a finite blocker.
    (c) In one dimension, persistent storage regions (orange) form columns extending from the initial to the final time boundary. The entanglement membrane (purple), which represents the entanglement between the final system and the initial reference, and hence the entropy remaining in the final state, must cross each column.
    (d) The same obstruction persists in two (or more) dimensions, where the membrane must again cross each storage column.
    (e) One period of a simple uniform Clifford circuit consisting of alternating $\hat Z$-measurement layers and two-qubit Clifford layers on even and odd bonds.
    (f) A factorized two-qubit Clifford gate provides a simple explicit blocker.
    (g) A finite unmeasured region with factorized gates on both boundaries provides a simple explicit storage region.}
    \label{fig:one-dimension}
\end{figure*}

Having established the plateau structure and the associated logical dynamics, we now turn to a broad class of spatially random one-dimensional stabilizer Floquet circuits. A circuit realization consists of a finite-depth pattern of local Clifford gates and Pauli measurements sampled from a distribution with only short-range spatial correlations. Once sampled, the entire pattern is fixed and repeated exactly in every Floquet period.

The ensemble is controlled by a measurement-rate parameter $p$, with increasing $p$ favoring configurations containing more measurements. Small $p$ corresponds to a regime deep in the mixed phase, while large $p$ corresponds to a regime deep in the area-law entangled phase. We consider families for which the support of the distribution on any finite spatial region is independent of $p$ throughout $0<p<1$. Equivalently, any finite configuration of gates and measurements that has nonzero probability at one value of $p$ in this interval has nonzero probability at every other value. Consequently, any finite local pattern that occurs with nonzero probability at one value of $p$ remains allowed throughout $0<p<1$, although its probability may become very small.

Because the sampled circuit is repeated exactly in time, any rare spatial configuration recurs every period and therefore persists indefinitely. These persistent rare regions can affect entanglement and purification in qualitatively different ways. We study two standard probes of measurement-induced transitions. The first starts from a pure product state and diagnoses the late-time bipartite entanglement across a spatial cut. The second starts from the maximally mixed state and probes purification through the decay of the entropy of the full system. In spacetime-random monitored circuits these two probes diagnose the same transition, whereas below we argue that persistent rare regions separate them for all intermediate values of $p$ in the class of one-dimensional Floquet circuits considered here.

\subsection{Area-Law Entanglement from Persistent Blockers}

We call a finite spatial region a \textit{blocker} if repeated Floquet evolution cannot produce entanglement across it. Near the high-measurement-rate limit, where the system is deep in the area-law phase, we expect the local circuit distribution to contain a blocking configuration supported on a finite spatial region. Since each Floquet period has finite depth, specifying such a configuration within one period involves only finitely many gate and measurement locations. For Clifford gates and Pauli measurements, only finitely many patterns are possible on these locations, so any blocking pattern in the support has nonzero probability. By the support condition, the same blocking pattern remains allowed throughout $0<p<1$. It may become increasingly rare as $p$ decreases, but for short-range-correlated disorder it still occurs at nonzero spatial density. Exact Floquet repetition makes each such blocker persist for all times.

In one dimension, a finite blocker separates the chain, as illustrated in Fig.~\ref{fig:one-dimension}(a). A nonzero density of persistent blockers therefore partitions the chain into finite segments. Entanglement can grow within each segment but cannot be produced across the blockers, so the late-time bipartite entanglement across a typical cut remains finite. Thus the system remains in the area-law phase for all $p>0$. If instead one considers the largest entanglement value over all cuts, the longest blocker-free segment grows as $O(\log L)$, so this maximum can grow at most logarithmically with $L$.

This mechanism is specific to one spatial dimension. In two or more dimensions, a finite blocker does not separate space, and paths for entanglement production can detour around it, as illustrated in Fig.~\ref{fig:one-dimension}(b). A nonzero density of blockers therefore does not by itself enforce area-law entanglement, except in one dimension.

\subsection{Purification Obstructed by Persistent Storage Regions}
\label{sec:one_dimension:purification}

To analyze purification under the circuit dynamics, we use the entanglement-membrane picture, in which the entanglement entropy of a pure state is represented by the free energy of an interface extending through spacetime~\cite{Nahum_Haah_entanglementgrowth_2016,Jonay2018,Zhou2019,Zhou2020,sierant_turkeshi_membrane2023}. We represent the maximally mixed initial state as a pure state in which each system qubit is maximally entangled with a corresponding external reference qubit. The entropy remaining after evolution is then the entanglement between the final system and the reference. In the corresponding statistical-mechanics description~\cite{jian_ludwig_prb2020,bao_choi_altman_prb2020}, these boundary conditions force a membrane to span the system in the spatial direction, separating the final system from the initial reference. The membrane may wander in spacetime, but it must still extend across the system~\cite{Li_Vijay_Fisher_DPRE2023,Ha2026}.

For stabilizer circuits, the entanglement membrane is effectively at zero temperature~\cite{sommers2024}. Its free energy is therefore set by the minimum membrane energy, with no contribution from the configurational entropy of the membrane. At finite temperature, configurational entropy can lower the membrane free energy at long times, whereas this mechanism is absent at zero temperature.

We call a finite spatial region a \textit{storage region} if its repeated Floquet dynamics preserves at least one logical qubit of the initial state. Near the low-measurement-rate limit, where the system is deep in the mixed phase, we expect the local circuit distribution to contain at least one storage pattern supported on a finite spatial region. The same local-support argument used for blockers then implies a nonzero density of such storage regions for all $p<1$, although this density may become very small as $p$ increases. Exact Floquet repetition makes each storage region persistent in time.

In spacetime, each persistent storage region appears as a column extending from the initial to the final time boundary, as illustrated in Fig.~\ref{fig:one-dimension}(c,d). The membrane energy associated with crossing a storage column equals the quantum information carried by that column~\cite{sommers2024}. For a Clifford circuit, this amount is an integer number of logical qubits, so crossing any storage region costs at least one bit of membrane energy. Any membrane separating the final state from the initial reference must cross every storage column, in any spatial dimension. A nonzero density of storage regions therefore makes the minimum membrane energy, and hence the late-time entropy, extensive. Thus the system remains in the mixed phase for all $p<1$ in any finite spatial dimension.

\subsection{A Simple Example}

As one concrete example of how blockers and storage regions can arise, consider a simple brickwork Clifford Floquet circuit in one spatial dimension.  One Floquet period consists of a $\hat Z$-measurement layer, a layer of independently sampled two-qubit Clifford gates on even bonds, a second measurement layer, and a Clifford layer on odd bonds, as shown in Fig.~\ref{fig:one-dimension}(e). Each two-qubit gate is sampled uniformly from the Clifford group, and each qubit is independently measured with probability $p$ in each measurement layer. For any $0<p<1$, every finite pattern of gates and measurement locations therefore has nonzero probability, so the local support is independent of $p$. The sampled pattern is then fixed and repeated exactly every period.

A two-qubit gate that factorizes into single-qubit Clifford rotations provides an explicit blocker, as shown in Fig.~\ref{fig:one-dimension}(f). Of the $11{,}520$ two-qubit Clifford gates, $24^2=576$ factorize into products of single-qubit Clifford gates, so a uniformly sampled gate factorizes with probability $1/20$~\cite{KoenigSmolin2014}. The identity is the trivial example. Such a gate cannot generate entanglement across that bond, giving a nonzero density of persistent blockers even at $p=0$. In one dimension, these blockers enforce area-law entanglement for all $p$. More general one- and two-sided blocking walls provide additional examples~\cite{farshi2022,farshi2023}.

As another simple example, a storage region can be formed by a finite segment whose two boundary gates factorize and whose qubits are never measured during the period, as shown in Fig.~\ref{fig:one-dimension}(g). The factorized boundary gates decouple the segment from the rest of the chain, while the absence of measurements leaves its internal evolution unitary. A maximally mixed state on the segment therefore retains its entropy indefinitely. For any fixed segment length and any $p<1$, this configuration has nonzero probability and hence occurs at nonzero density. These storage regions give an extensive lower bound on the late-time entropy, so the system remains in the mixed phase for all $p<1$.

\section{Two-Dimensional Entanglement Transition}
\label{sec:two_dimensional_transition}

We now turn to two spatial dimensions. The storage region mechanism that obstructs purification in one dimension remains effective. For any $p<1$, the Floquet pattern has a nonzero probability of containing a finite storage region, for example an unmeasured region enclosed by blocking boundaries. Such regions occur at nonzero density and thus maintain the volume-law entropy of the system after it is initialized in the maximally-mixed state. 
The long-time limit of the full-system entropy therefore remains extensive in the thermodynamic limit, and the purification dynamics stay in the mixed phase, consistent with the numerical results in Appendix~\ref{app:2d_purification}.

The entanglement dynamics are qualitatively different. In one dimension, a blocking region disconnects the chain and prevents entanglement from being produced across it.  In two (or more) dimensions, entanglement can be produced along paths that go around dilute finite local blockers.  This suggests that at low blocker density, connected paths allow entanglement to spread throughout the system and produce a volume-law phase, whereas at high blocker density, these paths are cut off and the late-time entanglement becomes area-law. Thus, despite the absence of a purification transition, we expect the two-dimensional Clifford Floquet circuit to exhibit an entanglement transition between volume-law and area-law phases. We numerically establish this transition below using bipartite entanglement and antipodal mutual information.

\subsection{Model}

We consider qubits on an $L_x\times L_y$ square lattice, with $L_x$ and $L_y$ even. One Floquet period $\mathcal{C}_F$ consists of four layers of nearest-neighbor two-site Clifford unitaries, each followed by a measurement layer. In each measurement layer, every site is selected independently with probability $p$ and measured in the $\hat Z$ basis. The four measurement patterns and all Clifford gates are sampled once and then repeated in every Floquet period.

For each $\boldsymbol{\delta}\in\{\hat{x},\hat{y},-\hat{x},-\hat{y}\}$, in that order, one layer applies an independently and uniformly sampled two-site Clifford unitary to each pair of sites $(i,j)$ and $(i,j)+\boldsymbol{\delta}$ with $i+j$ even. Together, these four layers form a two-dimensional analogue of the brickwork circuit. Unless otherwise stated, we use periodic boundary conditions in both directions. Cylindrical boundary conditions are used for the bipartite-entanglement calculation below.

All pure-state entanglement ``observables'' are evaluated after a number of Floquet periods chosen to be safely beyond both the plateau time and the entanglement-saturation time. When an observable continues to oscillate at late times, we average it over the final few Floquet periods.

\subsection{Bipartite Entanglement}

\begin{figure}[t]
    \centering
    \includegraphics[width=0.5\textwidth]{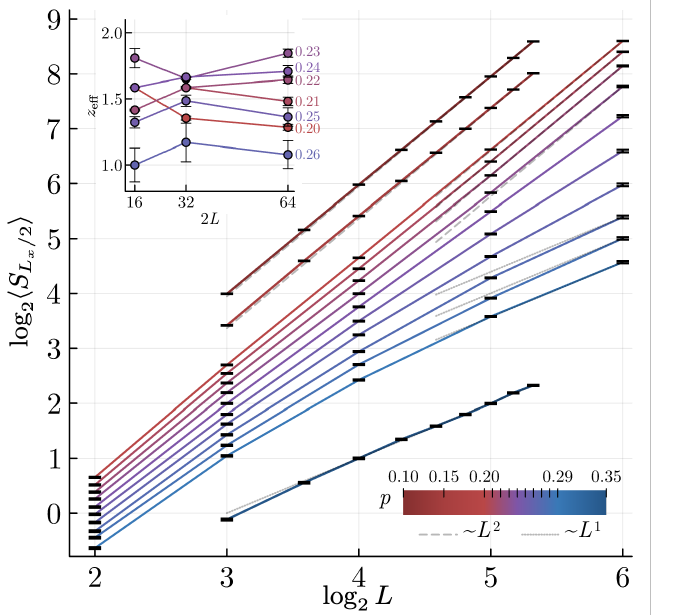}
    \caption{Bipartite entanglement in the two-dimensional Floquet monitored Clifford circuit. The main panel shows $\log_2\langle S_{L_x/2}\rangle$ versus $\log_2 L$ for square systems $L_x=L_y=L$ at different measurement rates $p$, where $\langle\cdots\rangle$ denotes an average over circuit realizations. Gray dashed and dotted lines indicate $L^2$ and $L$ scaling, respectively. The inset shows the effective dynamical exponent $z_{\mathrm{eff}}$ obtained from the median saturation time for measurement rates near the transition. }
    \label{fig:2d_bipartite}
\end{figure}

We first diagnose the transition using the late-time entanglement starting from a pure product state. We consider an $L_x\times L_y$ system with open boundary conditions in the $x$ direction and periodic boundary conditions in the $y$ direction, and measure the bipartite entanglement entropy $S_{L_x/2}$ across the cut at $x=L_x/2$. Deep in the two phases, the averaged saturated entropy is expected to scale as
\begin{align}
\langle S_{L_x/2}\rangle \sim
\begin{cases}
L_xL_y, & \text{volume-law phase},\\
L_y, & \text{area-law phase}.
\end{cases}
\label{eq:2d_midcut_scaling}
\end{align}

For square systems with $L_x=L_y=L$, Fig.~\ref{fig:2d_bipartite} shows a clear crossover from approximately $L^2$ scaling at low measurement rates to approximately $L$ scaling at high measurement rates, establishing volume-law and area-law phases on the two sides of the transition. The crossover occurs near $p\simeq0.23$--$0.25$, although the bipartite entropy alone does not determine $p_c$ precisely at the accessible system sizes.

We next examine the time required for the entanglement to reach its late-time plateau. For each circuit realization and time, the entanglement entropy takes integer values and at long times reaches an identifiable plateau, up to small $O(1)$ temporal fluctuations.
This allows us to assign a saturation time to each realization. Representative entanglement trajectories and the precise definition of the saturation time are given in Appendix~\ref{app:entanglement_growth}. We denote the median saturation time over circuit realizations by $t_{\mathrm{sat}}(L;p)$ and define
\begin{equation}
z_{\mathrm{eff}}(L,2L;p)=\frac{\log t_{\mathrm{sat}}(2L;p)-\log t_{\mathrm{sat}}(L;p)}{\log 2}~.
\label{eq:zeff}
\end{equation}

At an infinite-randomness critical point, activated dynamical scaling leads to an effective dynamical exponent that grows with system size and diverges in the thermodynamic limit~\cite{zabalo2023}. For measurement rates near the transition, however, $z_{\mathrm{eff}}$ remains below $2$ over the accessible sizes and exhibits only weak size dependence, as shown in the inset of Fig.~\ref{fig:2d_bipartite}. We therefore find no indication of infinite-randomness dynamical slowing within the system sizes accessible to our simulations.

\subsection{Antipodal Mutual Information}

\begin{figure}[t]
    \centering
    \includegraphics[width=0.5\textwidth]{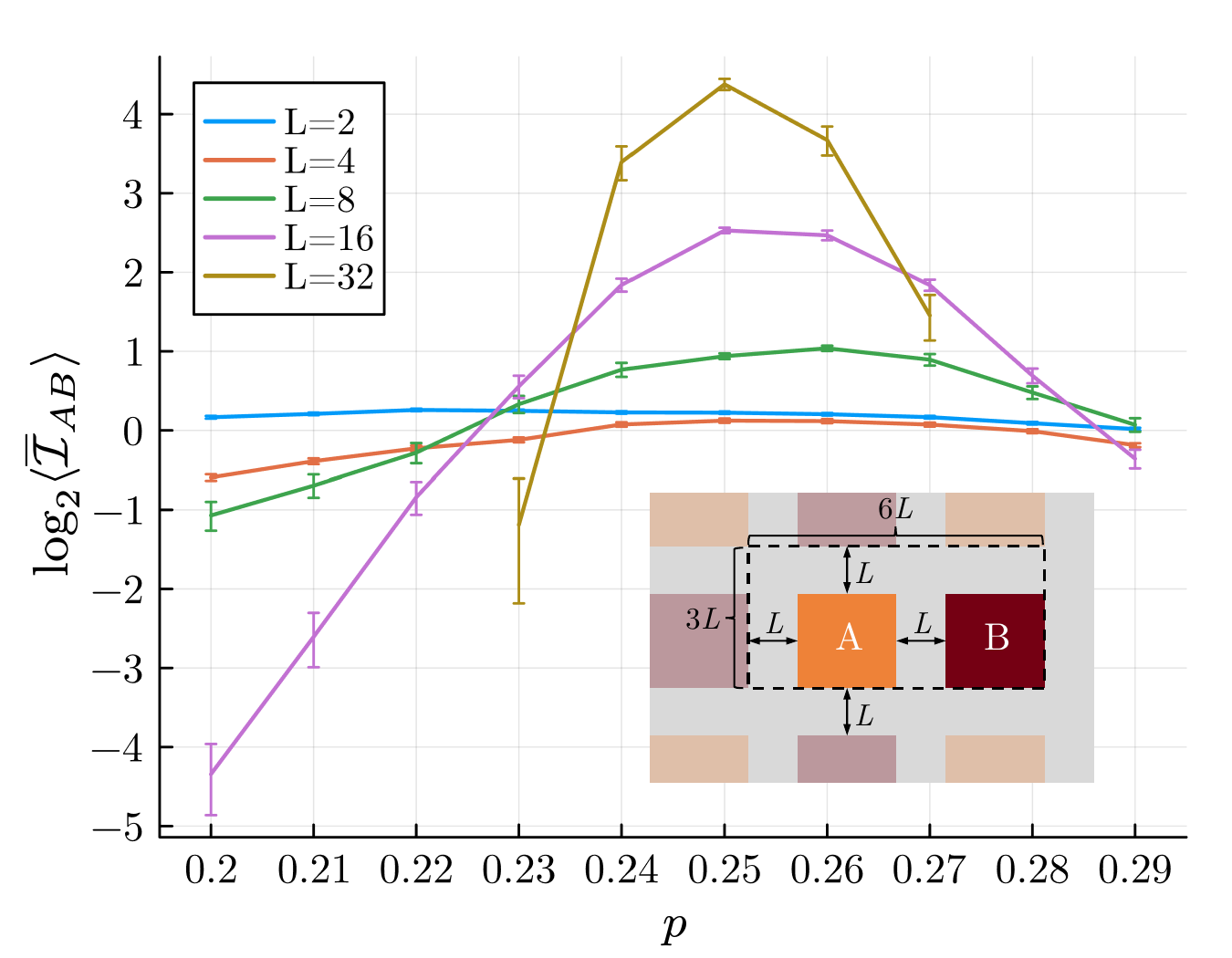}
    \caption{Antipodal mutual information in the two-dimensional Floquet monitored Clifford circuit. The main panel shows late-time-averaged $\log_2\langle\overline{\mathcal I}_{AB}\rangle$ as a function of measurement rate $p$ for different system sizes. A peak develops near the entanglement transition and grows strongly with system size, with a small finite-size drift in its position. The inset illustrates the AMI geometry: the regions $A$ and $B$ are $2L\times2L$ squares separated by a distance $L$ on all sides, and the dashed $3L\times6L$ rectangle defines the torus and the shifted periodic boundary conditions used.}
    \label{fig:AMI}
\end{figure}

To probe long-range entanglement, we consider the antipodal mutual information between two separated regions $A$ and $B$,
\begin{align}
\mathcal{I}_{AB}
=
S_A+S_B-S_{A\cup B}~.
\label{eq:2d_AMI}
\end{align}
We initialize the system in a pure product state and evolve it to times beyond the bipartite-entanglement saturation time for all sampled circuit realizations. We obtain the late-time AMI by averaging $I_{AB}(t)$ over the final few Floquet periods and denote this average by $\overline{\mathcal{I}}_{AB}$.
The AMI vanishes in the thermodynamic limit deep in both the volume-law and area-law phases, while finite systems develop a peak near the transition~\cite{LCF2019,ha2024}.

The geometry is shown in the inset of Fig.~\ref{fig:AMI}. The regions $A$ and $B$ are $2L\times2L$ squares separated by a distance $L$ on all sides. The primitive rectangle defining the torus has dimensions $3L\times6L$. Their combined size satisfies
\begin{align}
|A|+|B|
=
8L^2
<
\frac{18L^2}{2},
\label{eq:AMI_region_size}
\end{align}
so the two regions together occupy less than half of the system. This geometry is a two-dimensional analogue of the standard one-dimensional antipodal construction.

The numerical results are shown in Fig.~\ref{fig:AMI}. A pronounced peak develops near the entanglement transition and grows rapidly with system size. This is strikingly different from standard spacetime-random measurement-induced transitions, where the critical AMI approaches an $O(1)$ value, as discussed in Appendix~\ref{app:spacetime_random_AMI}. The observed growth is consistent with $L^{\alpha'}$ with $1<\alpha'<2$, although the accessible sizes do not determine the asymptotic scaling form. The peak position also drifts with system size, which prevents a precise determination of $p_c$ from the AMI.

Together with the saturation-time analysis, these results distinguish the present transition from both familiar possibilities. We find no indication of the strong dynamical slowing associated with infinite-randomness criticality, while the AMI scaling differs sharply from that of the standard spacetime-random MIPT.

\section{Localized and Delocalized Logical Dynamics}
\label{sec:deep_phases}

The plateau stabilizer group identifies the logical degrees of freedom that survive indefinitely, but does not determine their spatial dynamics. We now ask whether an initially local logical operator remains confined near its initial position or spreads over distances comparable to the system size.

\subsection{Logical-Support Radius}
\label{sec:logical_support_radius}

We probe these dynamics using logical Pauli operators that admit a representative supported on a single physical site. A Pauli operator $\hat P_r\in\{\hat X_r,\hat Y_r,\hat Z_r\}$ at site $r$ represents a single-site logical operator (SSLO) when
\begin{align}
    \hat P_r\in\mathrm{ISG}_\infty^\perp~,
    \qquad
    \hat P_r\notin\mathrm{ISG}_\infty~.
    \label{eq:SSLO_condition}
\end{align}

Let $[\hat P_r(t)]$ denote the logical operator obtained by evolving the initial SSLO $[\hat P_r]$ for $t$ Floquet periods. As established in Sec.~\ref{sec:logical_floquet_dynamics}, this evolution is generated by the logical Floquet unitary:
$[\hat P_r(t)]=[\mathcal{U}_F^t\hat P_r\mathcal{U}_F^{-t}]$.
Physical representatives of the same logical operator differ by multiplication of plateau stabilizers and can therefore have different supports.

The conventional code distance is the minimum weight among all nontrivial logical operators. For a specified logical operator, finding the minimum-weight physical representative is generally difficult. In one dimension, the \emph{contiguous code length} provides a geometric alternative based on the shortest contiguous interval that can support a representative of the logical operator~\cite{gullans2020,Bravyi2009,li2021qec}. Motivated by this construction, we introduce a spatial measure for individual logical operators that extends naturally to higher dimensions.

Consider a system of linear size $L$ with periodic boundary conditions in both spatial directions. For an SSLO initially located at site $r$, we define the ball of radius $R$ in the periodic Manhattan metric as
\begin{align}
    B_R(r)
    =
    \left\{
    r' :
    \sum_{\mu=x,y}
    \min\!\left(
    |r'_\mu-r_\mu|,
    L-|r'_\mu-r_\mu|
    \right)
    \le R
    \right\}.
    \label{eq:logical_ball}
\end{align}
Figure~\ref{fig:logical_support}(a) illustrates this geometry.
We define the \emph{logical-support radius} as the smallest $R$ for which some representative $\hat P_r(t)\hat S$, with $\hat S\in\mathrm{ISG}_\infty$, is supported entirely within $B_R(r)$. The center remains fixed at the initial site $r$, so a large radius means that no representative can be confined to a small neighborhood of that site. By construction, an SSLO has radius zero at $t=0$. Appendix~\ref{app:logical_support_radius} describes the numerical calculation.

\begin{figure}[t]
    \centering
    \includegraphics[width=0.5\textwidth]{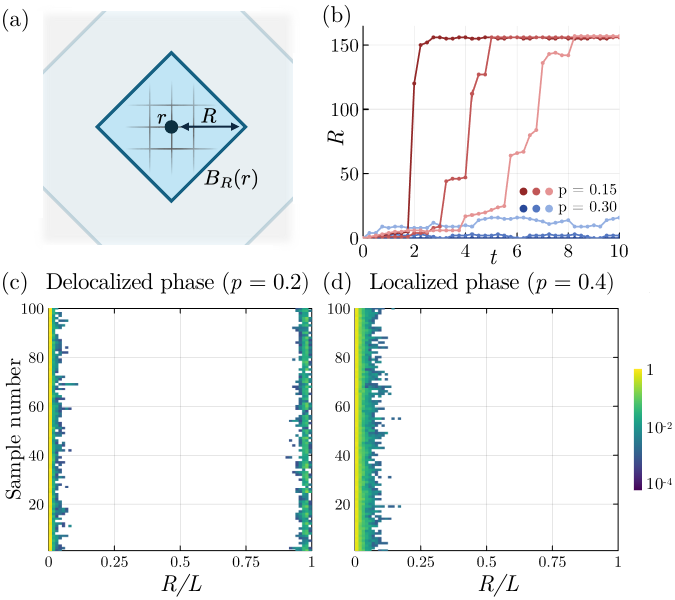}
    \caption{Logical-support radius in the localized and delocalized regimes.
    (a) Schematic of the definition. The ball $B_R(r)$ is centered at the initial site $r$ of the SSLO. The logical-support radius is the smallest $R$ for which the evolved logical operator has a physical representative supported within $B_R(r)$, allowing multiplication by plateau stabilizers. The lighter diamond illustrates a ball with radius comparable to the system size.
    (b) Evolution of the logical-support radius for selected SSLOs in two $L=160$ circuit realizations, one at each measurement rate. At $p=0.15$, the three red curves show delocalized SSLOs that reach system-scale radii. 
    Their growth occurs through abrupt jumps separated by intervals of slower evolution. At $p=0.3$, the three blue curves show localized SSLOs whose radii remain small.
    (c) Distribution of the late-time normalized radius $R/L$ for $L=240$ and $p=0.2$. Each row corresponds to one circuit realization and shows the probability distribution over its SSLOs at one late time. 
    Weight appears both near $R/L=0$ and near the system scale $R/L=1$.
    (d) Corresponding distributions for $L=120$ and $p=0.4$, where the radii remain concentrated near $R/L=0$. In (c) and (d), color denotes the probability on a logarithmic scale.}
    \label{fig:logical_support}
\end{figure}

\subsection{Localized and Delocalized Logical Support}
\label{sec:localized_delocalized_support}

For each disorder realization, we identify all SSLOs and follow their logical-support radii in time. After the initial transient, temporal fluctuations are typically small over the observation window. We therefore characterize each SSLO by its late-time radius and construct the distribution of these radii within each realization.

Deep in the volume-law phase [Fig.~\ref{fig:logical_support}(c)], these distributions are often bimodal, with weight at both small and system-wide radii. Localized and delocalized SSLOs thus coexist within individual disorder realizations. For the delocalized SSLOs, no physical representative of the evolved logical operator can be confined to a small neighborhood of its initial site, except at exponentially late and rare ``recurrence'' times. Deep in the area-law phase [Fig.~\ref{fig:logical_support}(d)], the late-time radii remain concentrated at small values, with no comparable population at the system scale.

To resolve how these late-time radii are reached, we calculate the radius after each of the four substeps of a Floquet period, giving a time resolution of one quarter period. In the localized phase, the radius remains small throughout the observed evolution. Fig.~\ref{fig:logical_support}(b) shows some example trajectories in the delocalized phase, where the radius eventually reaches the scale of the entire system. The radius shows abrupt, step-like increases rather than smooth ballistic growth. This differs from the operator spreading typically found in generic local unitary dynamics without measurements. These abrupt expansions make the spreading resemble teleportation~\cite{bao2024,Shkolnik2025}.

\section{Logical space localization transition}
\label{sec:logical_krylov}

\begin{figure*}[t]
    \centering
    \includegraphics[width=\textwidth]{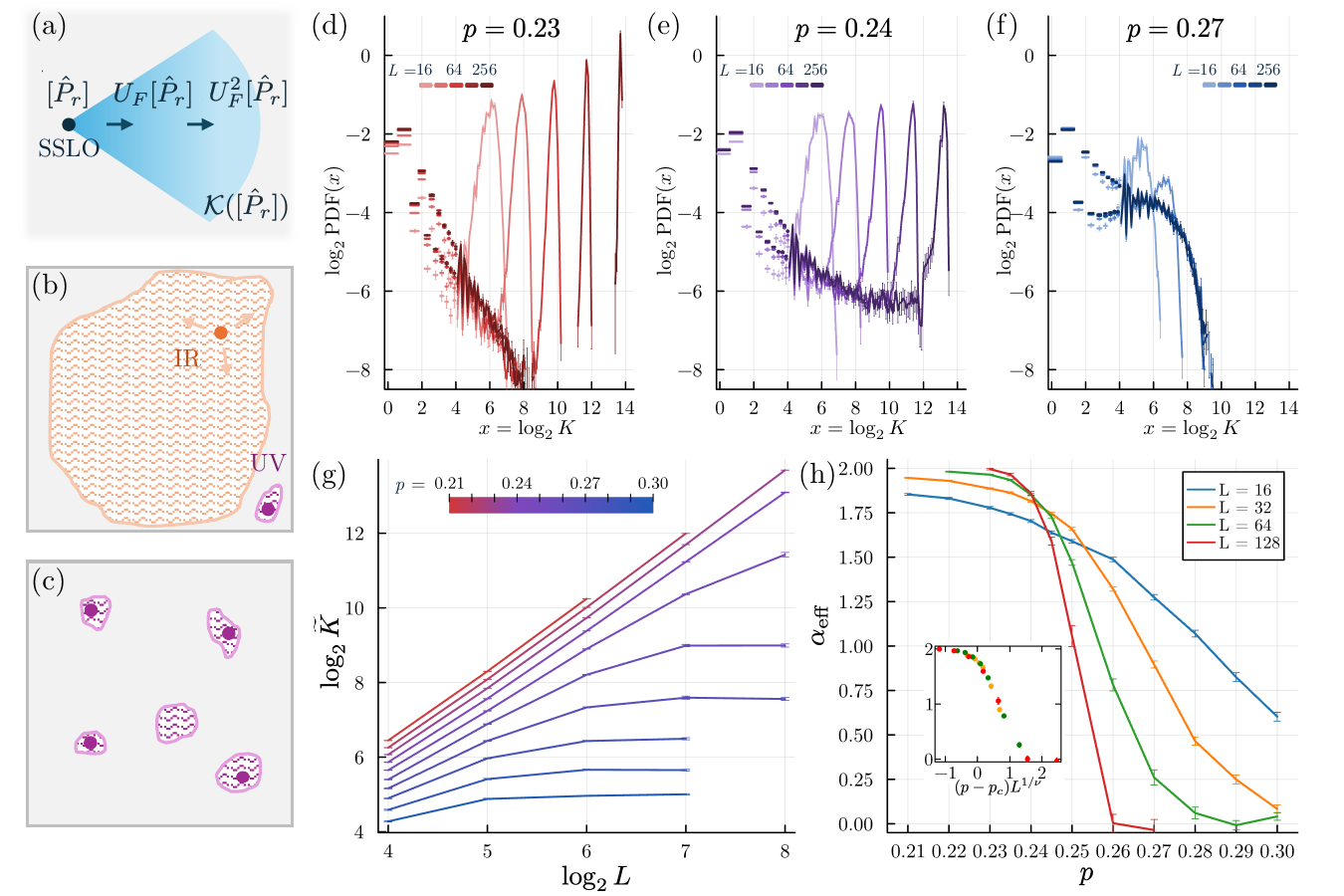}
    \caption{Logical space localization transition from Krylov dynamics. (a) Repeated Floquet evolution of a SSLO generates a logical Krylov subspace $\mathcal{K}([\hat P_r])$ of dimension $K$. (b),(c) Schematic of the delocalized and localized phases. Dots denote SSLOs and shaded regions illustrate the logical degrees of freedom explored by their evolution. The delocalized phase contains both finite UV contributions and a growing IR contribution. The localized phase contains only finite UV contributions. Logical subspaces containing no SSLO are not directly sampled by this probe. (d)--(f) Krylov dimension distributions at $p=0.23$, $0.24$, and $0.27$. The component at large $K$ shifts to larger values with increasing $L$ in the delocalized phase and near the transition. In the localized phase the distribution becomes independent of system size. The histograms show the probability density of $x=\log_2K$ using bins of width $0.1$. These bins resolve individual integer values of $K$ at small $K$. (g) The moment ratio $\widetilde K=\langle K^2\rangle/\langle K\rangle$ approaches extensive scaling in the delocalized phase and saturates in the localized phase. (h) Effective exponent $\alpha_{\mathrm{eff}}(L,2L;p)$ with $L$ denoting the smaller size in each pair. The curves flow toward $2$ at lower $p$ and toward $0$ at higher $p$. Their size dependence is weak near $p\simeq0.24$. The inset shows a scaling collapse of the three largest size pairs using $p_c=0.243$ and $\nu=1.1$.} 
    \label{fig:krylov_main}
\end{figure*}

\subsection{Logical Krylov space and distributions}
\label{sec:krylov_distribution}

The logical support radius introduced in Sec.~\ref{sec:deep_phases} distinguishes logical operators that remain near their initial positions from those that spread over distances comparable to the system size. We now characterize the logical degrees of freedom explored by this evolution through the number of independent logical Pauli operators generated from a local seed SSLO.

Let $\hat P_r$ be a SSLO and $[\hat P_r]$ its equivalence class in the surviving logical Pauli space $\mathcal{L}_\infty$. Its logical Krylov subspace is
\begin{equation}
    \mathcal{K}([\hat P_r])=\mathrm{span}_{\mathbb{F}_2}\left\{[\hat P_r],\mathcal{U}_F[\hat P_r],\mathcal{U}_F^2[\hat P_r],\ldots\right\},
    \label{eq:logical_krylov_subspace}
\end{equation}
with dimension $K(\hat P_r)=\dim_{\mathbb{F}_2}\mathcal{K}([\hat P_r])$. Note that this is not the usual operator Krylov construction~\cite{parker2019, dymarsky2020, nandy2025}.  Instead, we are viewing the logical operators that are produced by the dynamics as the generators of a logical subspace.  Thus $K$ counts the dimension of the logical subspace that can be dynamically generated from $\hat P_r$ under repeated Floquet evolution and by taking products of the operators so-generated. Equivalently, if $v\in\mathbb{F}_2^{2k}$ represents $[\hat P_r]$ and $M_F$ represents the logical Floquet action, then $\mathcal{K}(v)=\mathrm{span}_{\mathbb{F}_2}\{v,M_Fv,M_F^2v,\ldots\}$. Since $\dim_{\mathbb{F}_2}\mathcal{L}_\infty=2k$ we have $K\leq2k$. Fig.~\ref{fig:krylov_main}(a) illustrates this construction.

If the logical support radius of an SSLO remains bounded independently of $L$ throughout its evolution then only a finite number of independent logical Pauli operators can be dynamically generated from it. Operators that spread across the system can instead generate a number that grows with $L$. We use UV and IR to denote contributions associated with finite and growing scales in both the spatial dynamics and the Krylov distribution. The coexistence of confined and spreading operators in the delocalized phase suggests both UV and IR contributions [Fig.~\ref{fig:krylov_main}(b)]. Logical space localization corresponds to the disappearance of the delocalized IR contribution [Fig.~\ref{fig:krylov_main}(c)].

We compute the distribution $P(K;L,p)$ by first normalizing the distribution of $K$ over the SSLOs within each disorder realization and then averaging over realizations. Each realization containing at least one SSLO carries equal statistical weight regardless of its number of SSLOs. Realizations containing no SSLO are excluded. Their frequency decreases rapidly with system size and none are observed for sufficiently large systems. This distribution characterizes the logical dynamics accessible from SSLOs.

In the delocalized phase [Fig.~\ref{fig:krylov_main}(d)] the distribution is strongly bimodal. The component at small $K$ approaches a distribution that is independent of system size. The peak at large $K$ shifts to larger values and becomes narrower relative to its position. These features distinguish the localized UV contribution from the delocalized IR contribution.

Near the transition [Fig.~\ref{fig:krylov_main}(e)] the bimodal structure persists. The characteristic values of $K$ in the component at large $K$ grow more slowly than $L^2$ over the accessible sizes. Its shape in the logarithmically-binned representation changes little for the largest sizes. This is consistent with its width and characteristic Krylov scale having the same size dependence.

In the localized phase [Fig.~\ref{fig:krylov_main}(f)] the delocalized IR component disappears. The full distribution approaches a distribution that is independent of system size $L$ and has an approximately exponential tail. The characteristic Krylov dimensions therefore remain finite.

We describe these observations through the phenomenological decomposition
\begin{equation}
    P(K;L,p)
    \simeq [1-W]P_{\mathrm{UV}}(K)
    +W P_{\mathrm{IR}}(K)~,
    \label{eq:krylov_uv_ir}
\end{equation}
where $W\equiv W(L,p)$ is the probability weight of the IR component. Both component distributions are normalized and their dependence on $L$ and $p$ is left implicit. Where the two components coexist both $W$ and $1-W$ remain of order unity as $L$ increases. The first two UV moments remain $O(1)$. We do not assume a specific shape for $P_{\mathrm{UV}}$.

We describe the IR component using the phenomenological form
\begin{equation}
    P_{\mathrm{IR}}(K;L,p)
    \simeq\frac{1}{\sigma_K(L,p)}
    f_{\mathrm{IR}}\!\left(
        \frac{K-K_*(L,p)}{\sigma_K(L,p)}
    \right),
    \label{eq:krylov_ir_distribution}
\end{equation}
where $K_*$ and $\sigma_K$ are the mean and standard deviation of the IR component. The scaling function $f_{\mathrm{IR}}$ is normalized with zero mean and unit variance. Its shape likely differs between the delocalized and critical regimes.  In the delocalized phase we expect $K_*\sim L^2$ while $\sigma_K\sim L$. The IR distribution therefore becomes narrow relative to its mean.  At the transition the data are consistent with $K_*\sim\sigma_K\sim L^\alpha$ with $\alpha<2$ over the accessible sizes. Separating the UV and IR components becomes more difficult in this regime. We therefore seek a measure of the growing IR scale that can be computed directly from the full distribution.

\subsection{Critical behavior at the logical space localization transition}
\label{sec:krylov_transition}

We define the moment ratio
\begin{equation}
    \widetilde K(L,p)
    \equiv\frac{\langle K^2\rangle}{\langle K\rangle},
    \label{eq:krylov_moment_ratio}
\end{equation}
where the moments are taken with respect to $P(K;L,p)$. For sufficiently large systems the IR component dominates both moments provided that its weight $W$ does not decrease too rapidly with system size. Near the transition, this condition is supported numerically by the growth of $\langle K\rangle$ with $L$, as shown in Appendix~\ref{app:krylov_mean} and Fig.~\ref{fig:krylov_mean}.
The IR weight cancels to leading order in the ratio so that $\widetilde K\sim K_*$. This holds both in the delocalized phase where $\sigma_K\ll K_*$ and near the transition where $\sigma_K\sim K_*$. Thus $\widetilde K$ gives a measure of the values of $K$ in the IR part of the distribution without requiring an explicit separation of the UV and IR components.  In the localized phase both moments remain finite and $\widetilde K$ saturates.

Fig.~\ref{fig:krylov_main}(g) shows $\log_2\widetilde K$ versus $\log_2L$. In the delocalized phase $\widetilde K$ approaches the extensive scaling $\widetilde K\sim L^2$. In the localized phase it saturates with increasing system size. Near the transition $\widetilde K$ continues to grow over the accessible sizes with an exponent smaller than $2$.

We quantify this growth using the effective exponent between successive sizes
\begin{equation}
    \alpha_{\mathrm{eff}}(L,2L;p)
    =\log_2\!\left[
        \frac{\widetilde K(2L,p)}{\widetilde K(L,p)}
    \right].
    \label{eq:krylov_alpha_eff}
\end{equation}
If $\widetilde K\sim L^\alpha$ then $\alpha_{\mathrm{eff}}\to\alpha$. The delocalized and localized phases therefore correspond to limiting values of $2$ and $0$.

Fig.~\ref{fig:krylov_main}(h) shows how $\alpha_{\mathrm{eff}}$ evolves with increasing system size. It increases toward $2$ for $p<0.24$ and decreases toward $0$ for $p>0.24$. The pair $(L,2L)=(16,32)$ shows strong finite-size effect. At $p=0.24$ the three largest size pairs give nearly the same value $\alpha_{\mathrm{eff}}\approx1.85$. These trends suggest that the transition is near $p_c\approx0.24$.

To estimate the critical point and the correlation length exponent we consider the scaling form
\begin{equation}
    \alpha_{\mathrm{eff}}(L,2L;p)
    \simeq\mathcal{G}\!\left[(p-p_c)L^{1/\nu}\right].
    \label{eq:krylov_alpha_scaling}
\end{equation}
The inset of Fig.~\ref{fig:krylov_main}(h) shows a collapse of the three largest size pairs over $0.23\leq p\leq0.27$ using $p_c=0.243$ and $\nu=1.1$. The data do not determine $\nu$ precisely and a broad range of values give comparable collapse quality, so this only gives a preliminary characterization of its critical scaling. The displayed value is compatible with the Harris criterion $\nu\geq1$ for disorder that is random in two spatial dimensions~\cite{harris1974}.

The surviving logical space $\mathcal{L}_\infty$ remains extensive on both sides of the transition. In the volume-law phase the IR Krylov scale grows as $L^2$ while finite UV scales persist. In the area-law phase the growing IR contribution is absent and the characteristic Krylov dimensions remain finite. Together with the spatial localization and delocalization established in Sec.~\ref{sec:deep_phases}, these results identify the entanglement transition as a localization transition within the surviving logical space.

\section{Discussion}
\label{sec:discussion}

Repeating the same Clifford unitaries and Pauli measurements every period produces a striking separation between entanglement and purification. In conventional spacetime-random monitored circuits, the two probes diagnose the same measurement-induced phase transition. Here, the repeated circuit makes rare spatial structures persistent. In one dimension, persistent blockers enforce area-law entanglement, while isolated weakly-measured regions retain a finite entropy density. In two (or more) dimensions, entanglement can be produced along paths that avoid local blockers, producing a volume-law to area-law entanglement transition, while the late-time entropy remains extensive for every $p<1$ if the system is initialized in the maximally mixed state, so there is no purification transition.

The entanglement critical point itself also shows unusual behavior. In the corresponding spacetime-random circuit, the AMI peak remains below one over the accessible sizes and is consistent with an $O(1)$ critical value [e.g., Fig.~\ref{fig:AMI_random}]. In the Floquet circuit, the peak instead grows rapidly with $L$ and reaches values much larger than one. The entanglement saturation time provides a complementary dynamical probe. The corresponding effective dynamical exponent $z_{\mathrm{eff}}$ shows only weak size dependence near the transition, giving no indication of the rapidly growing dynamical exponent associated with infinite-randomness criticality.

The plateau stabilizer group and the induced logical Floquet map provide a natural interpretation of these results. The plateau stabilizer group specifies the stabilizer constraints imposed by the repeated monitored circuit, while its associated logical space contains the information that remains unpurified. The logical Floquet map then determines how these surviving logical degrees of freedom evolve. In two (or more) dimensions, the dimension of the surviving logical space remains extensive across the entanglement transition, while its dynamics change from delocalized to localized. The transition can therefore be viewed as a localization transition within an extensive surviving logical space.

One possible challenge for future work about this phase transition is to try to find a model and/or other metrics with substantially weaker finite-size effects, so that the critical behavior such as exponents can be more precisely determined.

An important question is how this picture changes away from the Clifford limit. At the Clifford point, entropies are integer-valued and the evolution closes exactly on stabilizer groups. Together with exact repetition of the same circuit, this discrete structure permits perfectly blocking regions and exactly protected logical information. A generic non-Clifford perturbation removes this property. Entropy can then decrease by arbitrarily small amounts, while entanglement can increase by arbitrarily small amounts across a weakly broken blocker. The strictly nonpurifying plateau may therefore remain long lived without being permanent. For the same reason, the exact blocker argument for area-law entanglement in one dimension no longer applies.

The logical-space scaling suggests a possible form for the resulting purification time. In the volume-law phase, the dynamically connected logical sector is extensive and scales as $\sim L^2$ in two dimensions.  If a weak non-Clifford (``magic'') perturbation produces a leakage time that is exponential in the size of this sector, one would obtain
\begin{align}
    t_{\mathrm{pur}}\sim \exp(cL^2)~.
\end{align}
Near the transition, our Krylov data instead suggest a subextensive scale
$\widetilde K\sim L^\alpha$ with $\alpha<2$.
If this logical-sector size controls the same leakage process, the purification time would scale as
\begin{align}
    t_{\mathrm{pur}}\sim \exp(cL^\alpha)~.
\end{align}

This last form is reminiscent of activated scaling at an infinite-randomness critical point, which is what this critical point might cross over to for such non-Clifford Floquet monitored systems.
The long purification time proposed above would arise from lifting the exact protection of the Clifford limit which has $t_{\mathrm{pur}}=\infty$. Whether this scaling survives generic non-Clifford perturbations remains an open question. It will also be interesting to determine whether the exponent $\alpha$ is related to the anomalous growth of the critical antipodal mutual information.

\begin{acknowledgments}

We thank Soonwon Choi, Sarang Gopalakrishnan, Michael Gullans, Nicholas O'Dea and Grace Sommers for discussions and for related collaborations.  This work was supported in part by NSF QLCI grant OMA-2120757. H.H. was supported in part by the Gordon and Betty Moore Foundation through Grant GBMF13896 to the MIT Quantum Initiative (QMIT). The simulations presented in this work were performed on computational resources managed and supported by Princeton Research Computing, a consortium of groups including the Princeton Institute for Computational Science and Engineering (PICSciE) and Research Computing at Princeton University. The open-source \texttt{QuantumClifford.jl} package was used to simulate Clifford circuits~\cite{QuantumClifford}. ChatGPT (OpenAI) was used to assist with the implementation and debugging of parts of the numerical code.

\end{acknowledgments}

\appendix

\section{Inclusion under Monitored Clifford Evolution}
\label{app:inclusion_property}

For completeness, we give a direct proof that monitored Clifford evolution preserves inclusion of stabilizer groups modulo Pauli phases. This property underlies the plateau construction reviewed in Sec.~\ref{sec:plateau_logical}.

Let $\mathcal{S}_1$ and $\mathcal{S}_2$ be stabilizer groups satisfying
\begin{align}
    \mathcal{S}_1 \subseteq \mathcal{S}_2.
\end{align}
It is sufficient to verify that each elementary operation in a monitored Clifford circuit preserves this inclusion.

A Clifford unitary maps Pauli operators by conjugation and therefore trivially preserves inclusion:
\begin{align}
    \mathcal{U}[\mathcal{S}_1]
    \subseteq
    \mathcal{U}[\mathcal{S}_2].
\end{align}

Now consider a measurement of a Pauli operator $P$. Working modulo stabilizer signs, the standard stabilizer measurement update is~\cite{gottesman1997,hastings2021,aasen2023}
\begin{align}
    \mathcal{M}_P[\mathcal{S}]
    =
    \left(\mathcal{S}\cap P^\perp\right)
    +
    \langle P\rangle,
\end{align}
where $P^\perp$ denotes the set of Pauli operators commuting with $P$, and $+$ denotes the subgroup generated by its two arguments. Since
\begin{align}
    \mathcal{S}_1\cap P^\perp
    \subseteq
    \mathcal{S}_2\cap P^\perp,
\end{align}
we have
\begin{align}
    \left(\mathcal{S}_1\cap P^\perp\right)+\langle P\rangle
    \subseteq
    \left(\mathcal{S}_2\cap P^\perp\right)+\langle P\rangle.
\end{align}
Thus,
\begin{align}
    \mathcal{M}_P[\mathcal{S}_1]
    \subseteq
    \mathcal{M}_P[\mathcal{S}_2].
\end{align}

Both Clifford unitaries and Pauli measurements therefore preserve inclusion. Since a monitored Clifford circuit $\mathcal{C}$ is a composition of these operations,
\begin{align}
    \mathcal{C}[\mathcal{S}_1]
    \subseteq
    \mathcal{C}[\mathcal{S}_2],
\end{align}
which establishes Eq.~\eqref{eq:inclusion_property}.

\section{Normal Form for Monitored Clifford Circuits}
\label{app:normal_form}

To construct the effective Floquet unitary, we express a monitored Clifford circuit as a Clifford unitary followed by a commuting stabilizer measurement. The equality is understood as an equality of maps on arbitrary input stabilizer groups modulo Pauli phases. Related canonical forms have been developed for stabilizer operations viewed as quantum channels and for postselected stabilizer computations~\cite{heimendahl2022,beverland2020}.

For a stabilizer group $\mathcal{S}$, $\mathcal{M}_{\mathcal{S}}$ denotes measurement of an independent generating set of $\mathcal{S}$. For $\mathcal{S}=\langle P\rangle$, we write $\mathcal{M}_P\equiv\mathcal{M}_{\langle P\rangle}$. The generators commute and may therefore be measured in any order. Measurement outcomes determine only the signs of the measured stabilizers and are ignored in our representation.

\begin{theorem}[Monitored Clifford Normal Form]
\label{thm:normal_form}
Let $\mathcal{C}$ be a circuit composed of Clifford unitaries and Pauli measurements. Then there exist an effective Clifford unitary $\mathcal{U}_{\mathrm{eff}}$ and stabilizer groups $\mathcal{S}_{\mathrm{i}}$ and $\mathcal{S}_{\mathrm{f}}$, related by
\begin{align}
    \mathcal{S}_{\mathrm{f}}
    =
    \mathcal{U}_{\mathrm{eff}}[\mathcal{S}_{\mathrm{i}}],
\end{align}
such that, for every input stabilizer group $\mathcal{S}_0$,
\begin{align}
    \mathcal{C}[\mathcal{S}_0]
    &=
    \left(
        \mathcal{M}_{\mathcal{S}_{\mathrm{f}}}
        \circ
        \mathcal{U}_{\mathrm{eff}}
    \right)[\mathcal{S}_0],
    \label{eq:unitary_measurement_normal_form}
    \\
    \mathcal{C}[\mathcal{S}_0]
    &=
    \left(
        \mathcal{U}_{\mathrm{eff}}
        \circ
        \mathcal{M}_{\mathcal{S}_{\mathrm{i}}}
    \right)[\mathcal{S}_0].
    \label{eq:measurement_unitary_normal_form}
\end{align}
The effective Clifford unitary $\mathcal{U}_{\mathrm{eff}}$ is generally not unique.
\end{theorem}

\subsection{Single Measurement Reduction}

We first prove the following lemma, which contains the main nontrivial step in the proof of Theorem~\ref{thm:normal_form}.

\begin{lemma}[Single Measurement Reduction]
\label{lem:single_measurement_reduction}
Let $P$ be a Pauli operator and let $\mathcal{S}$ be a stabilizer group. There exist a stabilizer group $\mathcal{S}'$ and a Clifford unitary $\mathcal{U}$, depending only on $P$ and $\mathcal{S}$, such that for every input stabilizer group $\mathcal{S}_0$,
\begin{align}
    \left(
        \mathcal{M}_P
        \circ
        \mathcal{M}_{\mathcal{S}}
    \right)[\mathcal{S}_0]
    =
    \left(
        \mathcal{M}_{\mathcal{S}'}
        \circ
        \mathcal{U}
    \right)[\mathcal{S}_0].
    \label{eq:single_measurement_reduction}
\end{align}
\end{lemma}

\begin{proof}
Motivated by the standard Gottesman stabilizer measurement rule~\cite{gottesman1998}, we consider three possibilities.

First, $P$ may already be an element of $\mathcal{S}$. Then the measurement $\mathcal{M}_P$ is redundant after $\mathcal{M}_{\mathcal{S}}$, and Eq.~\eqref{eq:single_measurement_reduction} holds with
\begin{align}
    \mathcal{S}'=\mathcal{S},
    \qquad
    \mathcal{U}=\mathcal{I}.
\end{align}

Second, $P$ may be a nontrivial logical Pauli operator of $\mathcal{S}$, so that $P\in\mathcal{S}^{\perp}\setminus\mathcal{S}$. The two commuting measurement layers can then be combined into a single stabilizer group measurement, giving
\begin{align}
    \mathcal{S}'
    =
    \mathcal{S}+\langle P\rangle,
    \qquad
    \mathcal{U}
    =
    \mathcal{I}.
\end{align}

It remains to consider the case in which $P$ anticommutes with at least one element of $\mathcal{S}$. Let $r=\operatorname{rank}(\mathcal{S})$. We may choose generators
\begin{align}
    \mathcal{S}
    =
    \langle g_1,g_2,\ldots,g_r\rangle
\end{align}
such that $g_1$ anticommutes with $P$, while $g_i$ commutes with $P$ for every $i>1$. To construct such a generating set, choose one generator $g_1$ that anticommutes with $P$. Any other generator $g_i$ that anticommutes with $P$ may be replaced by $g_1g_i$, which commutes with $P$. These replacements leave the generated stabilizer group unchanged.

Because Pauli operators are identified up to overall phase, we choose representatives satisfying $P^2=g_1^2=\hat{\mathbb{I}}$ and define
\begin{align}
    \mathcal{U}
    =
    \frac{P+g_1}{\sqrt{2}}.
    \label{eq:single_measurement_unitary}
\end{align}
This is a Clifford unitary that exchanges $P$ and $g_1$ while leaving the remaining generators unchanged,
\begin{align}
    \mathcal{U}[P]
    &=
    g_1,
    &
    \mathcal{U}[g_1]
    &=
    P,
    &
    \mathcal{U}[g_i]
    &=
    g_i
    \quad (i>1).
    \label{eq:single_measurement_unitary_action}
\end{align}
We take
\begin{align}
    \mathcal{S}'
    =
    \mathcal{U}[\mathcal{S}]
    =
    \langle P,g_2,\ldots,g_r\rangle.
    \label{eq:single_measurement_reduced_group}
\end{align}

We now show that the same $\mathcal{U}$ and $\mathcal{S}'$ work for every input stabilizer group $\mathcal{S}_0$. Let
\begin{align}
    \mathcal{S}_1
    =
    \mathcal{M}_{\mathcal{S}}[\mathcal{S}_0],
    \qquad
    m
    =
    \operatorname{rank}(\mathcal{S}_1).
\end{align}
Measuring $\mathcal{S}$ ensures that $\mathcal{S}\subseteq\mathcal{S}_1$. We may therefore choose generators
\begin{align}
    \mathcal{S}_1
    =
    \langle
        g_1,g_2,\ldots,g_r,
        h_{r+1},\ldots,h_m
    \rangle.
\end{align}
Every $h_j$ commutes with $g_1$. If $h_j$ anticommutes with $P$, we may replace it by $g_1h_j$, which leaves $\mathcal{S}_1$ unchanged and makes it commute with $P$. The generators may therefore be chosen so that every $h_j$ commutes with both $P$ and $g_1$.

In this basis, measuring $P$ replaces $g_1$ by $P$ and leaves all other generators unchanged. Therefore,
\begin{align}
    \mathcal{M}_P[\mathcal{S}_1]
    &=
    \langle
        P,g_2,\ldots,g_r,
        h_{r+1},\ldots,h_m
    \rangle
    \nonumber\\
    &=
    \mathcal{U}[\mathcal{S}_1].
\end{align}
Since this holds for every $\mathcal{S}_0$,
\begin{align}
    \mathcal{M}_P
    \circ
    \mathcal{M}_{\mathcal{S}}
    =
    \mathcal{U}
    \circ
    \mathcal{M}_{\mathcal{S}}.
    \label{eq:measurement_replaced_by_unitary}
\end{align}
Moving the stabilizer group measurement through the unitary gives
\begin{align}
    \mathcal{U}
    \circ
    \mathcal{M}_{\mathcal{S}}
    =
    \mathcal{M}_{\mathcal{U}[\mathcal{S}]}
    \circ
    \mathcal{U}
    =
    \mathcal{M}_{\mathcal{S}'}
    \circ
    \mathcal{U}.
\end{align}
Combining these identities proves Eq.~\eqref{eq:single_measurement_reduction}.
\end{proof}

\subsection{Proof of the Normal Form}
\begin{figure}[tb]
    \centering
    \includegraphics[width=0.45\textwidth]{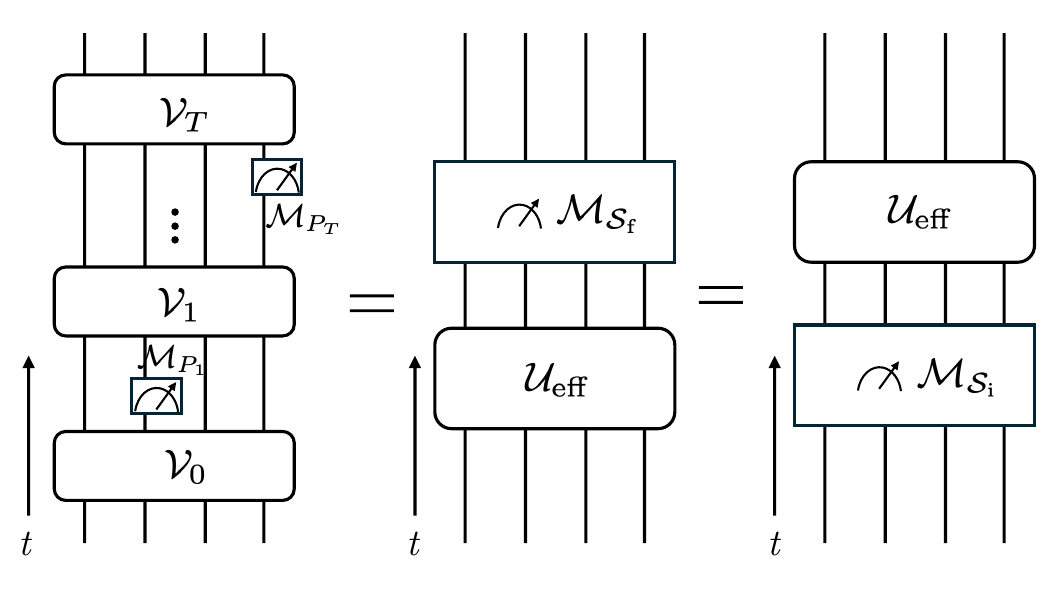}
    \caption{Normal-form representation of a monitored Clifford circuit. At the level of stabilizer groups modulo Pauli phases, the circuit can be represented by an effective Clifford unitary $\mathcal{U}_{\mathrm{eff}}$ followed by a single stabilizer-group measurement $\mathcal{M}_{\mathcal{S}_{\mathrm{f}}}$. Equivalently, the measurement can be applied before the effective Clifford unitary, with the initial and final measurement groups related by $\mathcal{S}_{\mathrm{f}}=\mathcal{U}_{\mathrm{eff}}[\mathcal{S}_{\mathrm{i}}]$. The final measurement group $\mathcal{S}_{\mathrm{f}}$ is the stabilizer group obtained by applying the circuit to the maximally mixed state.}
    \label{fig:normal_form}
\end{figure}

Consider a monitored Clifford circuit containing $T$ Pauli measurements,
\begin{align}
    \mathcal{C}
    =
    \mathcal{V}_T
    \circ
    \mathcal{M}_{P_T}
    \circ\cdots\circ
    \mathcal{V}_1
    \circ
    \mathcal{M}_{P_1}
    \circ
    \mathcal{V}_0,
    \label{eq:generic_monitored_clifford_circuit}
\end{align}
where each $\mathcal{V}_t$ is a Clifford unitary and each $P_t$ is a Pauli string. All Pauli operators and stabilizer groups are understood modulo overall phases, and circuit equalities refer to their induced maps on arbitrary input stabilizer groups.

A Clifford unitary can be moved through a Pauli measurement according to
\begin{align}
    \mathcal{U}
    \circ
    \mathcal{M}_P
    =
    \mathcal{M}_{\mathcal{U}[P]}
    \circ
    \mathcal{U}.
    \label{eq:move_measurement_through_unitary}
\end{align}
For a stabilizer group measurement, the corresponding identity is
\begin{align}
    \mathcal{U}
    \circ
    \mathcal{M}_{\mathcal{S}}
    =
    \mathcal{M}_{\mathcal{U}[\mathcal{S}]}
    \circ
    \mathcal{U}.
    \label{eq:move_stabilizer_measurement_through_unitary}
\end{align}

Using Eq.~\eqref{eq:move_measurement_through_unitary}, all Clifford unitaries in Eq.~\eqref{eq:generic_monitored_clifford_circuit} can be moved to the right. Denoting the resulting measured Pauli operators by $\widetilde P_1,\ldots,\widetilde P_T$, we obtain
\begin{align}
    \mathcal{C}
    =
    \mathcal{M}_{\widetilde P_T}
    \circ\cdots\circ
    \mathcal{M}_{\widetilde P_2}
    \circ
    \mathcal{M}_{\widetilde P_1}
    \circ
    \widetilde{\mathcal{U}}_1
\end{align}
for some Clifford unitary $\widetilde{\mathcal{U}}_1$.

We now absorb the measurement layers one at a time, starting from the right. Let
\begin{align}
    \mathcal{S}_1
    =
    \langle\widetilde P_1\rangle.
\end{align}
Applying Lemma~\ref{lem:single_measurement_reduction} to the next measurement gives a stabilizer group $\mathcal{S}_2$ and a Clifford unitary $\widetilde{\mathcal{U}}_2$ such that
\begin{align}
    \mathcal{M}_{\widetilde P_2}
    \circ
    \mathcal{M}_{\mathcal{S}_1}
    =
    \mathcal{M}_{\mathcal{S}_2}
    \circ
    \widetilde{\mathcal{U}}_2.
\end{align}
The circuit therefore becomes
\begin{align}
    \mathcal{C}
    =
    \mathcal{M}_{\widetilde P_T}
    \circ\cdots\circ
    \mathcal{M}_{\widetilde P_3}
    \circ
    \mathcal{M}_{\mathcal{S}_2}
    \circ
    \widetilde{\mathcal{U}}_2
    \circ
    \widetilde{\mathcal{U}}_1.
\end{align}

Repeating this step for each remaining measurement gives
\begin{align}
    \mathcal{M}_{\widetilde P_j}
    \circ
    \mathcal{M}_{\mathcal{S}_{j-1}}
    =
    \mathcal{M}_{\mathcal{S}_j}
    \circ
    \widetilde{\mathcal{U}}_j.
\end{align}
After the final measurement has been absorbed,
\begin{align}
    \mathcal{C}
    =
    \mathcal{M}_{\mathcal{S}_T}
    \circ
    \widetilde{\mathcal{U}}_T
    \circ\cdots\circ
    \widetilde{\mathcal{U}}_2
    \circ
    \widetilde{\mathcal{U}}_1.
\end{align}
With
\begin{align}
    \mathcal{S}_{\mathrm{f}}
    &\equiv
    \mathcal{S}_T,
    &
    \mathcal{U}_{\mathrm{eff}}
    &\equiv
    \widetilde{\mathcal{U}}_T
    \circ\cdots\circ
    \widetilde{\mathcal{U}}_1,
\end{align}
Eq.~\eqref{eq:unitary_measurement_normal_form} follows. Equation~\eqref{eq:measurement_unitary_normal_form} follows immediately from Eq.~\eqref{eq:move_stabilizer_measurement_through_unitary} by taking
\begin{align}
    \mathcal{S}_{\mathrm{i}}
    \equiv
    \mathcal{U}_{\mathrm{eff}}^{-1}
    [\mathcal{S}_{\mathrm{f}}].
\end{align}
This proves Theorem~\ref{thm:normal_form}.

For later use, the final measurement group is fixed by the action of $\mathcal{C}$ on the maximally mixed state. Since the trivial stabilizer group is invariant under any unitary,
\begin{align}
    \mathcal{C}[\langle\hat{\mathbb{I}}\rangle]
    &=
    \left(
        \mathcal{M}_{\mathcal{S}_{\mathrm{f}}}
        \circ
        \mathcal{U}_{\mathrm{eff}}
    \right)
    [\langle\hat{\mathbb{I}}\rangle]
    \nonumber\\
    &=
    \mathcal{S}_{\mathrm{f}}.
    \label{eq:final_measurement_group}
\end{align}

\section{Construction of the Effective Floquet Dynamics}
\label{app:logical_map}

We construct a fixed Clifford unitary $\mathcal{U}_F$ that preserves $\mathrm{ISG}_{\infty}$ and reproduces one Floquet period on every stabilizer group containing it. Its induced action on the logical Pauli space provides the map used for the logical diagnostics below.

For each $t$, let $\mathcal{U}_t$ denote an effective Clifford unitary in the normal form of the $t$ period circuit $\mathcal{C}_F^t$. Equation~\eqref{eq:final_measurement_group} gives
\begin{align}
    \mathcal{C}_F^t
    =
    \mathcal{M}_{\mathrm{ISG}_t}
    \circ
    \mathcal{U}_t,
    \qquad
    \mathrm{ISG}_t
    =
    \mathcal{C}_F^t[\langle\hat{\mathbb{I}}\rangle].
    \label{eq:t_period_normal_form}
\end{align}
For $t\geq t^*$,
\begin{align}
    \mathcal{C}_F^t
    =
    \mathcal{M}_{\mathrm{ISG}_\infty}
    \circ
    \mathcal{U}_t.
    \label{eq:post_plateau_normal_form}
\end{align}

\subsection{Construction of $\mathcal{U}_F$}

We first record a simple property of stabilizer group measurements.

\begin{lemma}
\label{lem:common_stabilizer_subgroup}
Let $\mathcal{A}$ and $\mathcal{R}$ be stabilizer groups, and let
\begin{align}
    \mathcal{B}
    =
    \mathcal{M}_{\mathcal{R}}[\mathcal{A}].
\end{align}
Then
\begin{align}
    \mathcal{A}\cap\mathcal{B}
    =
    \mathcal{A}\cap\mathcal{R}^{\perp}.
    \label{eq:common_stabilizer_subgroup}
\end{align}
\end{lemma}

\begin{proof}
The measurement update gives
\begin{align}
    \mathcal{B}
    =
    \left(
        \mathcal{A}\cap\mathcal{R}^{\perp}
    \right)
    +
    \mathcal{R},
\end{align}
which implies
$\mathcal{A}\cap\mathcal{R}^{\perp}
\subseteq
\mathcal{A}\cap\mathcal{B}$.
Conversely, if $P\in\mathcal{A}\cap\mathcal{B}$, then $P$ commutes with $\mathcal{R}\subseteq\mathcal{B}$ because $\mathcal{B}$ is a stabilizer group. Hence
$P\in\mathcal{A}\cap\mathcal{R}^{\perp}$.
\end{proof}

The one period normal form is
\begin{align}
    \mathcal{C}_F
    =
    \mathcal{M}_{\mathrm{ISG}_1}
    \circ
    \mathcal{U}_1.
    \label{eq:one_period_normal_form}
\end{align}
For $t\geq t^*$,
\begin{align}
    \mathcal{C}_F^{t+1}
    &=
    \mathcal{C}_F
    \circ
    \mathcal{C}_F^t
    \nonumber\\
    &=
    \left(
        \mathcal{M}_{\mathrm{ISG}_1}
        \circ
        \mathcal{U}_1
    \right)
    \circ
    \left(
        \mathcal{M}_{\mathrm{ISG}_\infty}
        \circ
        \mathcal{U}_t
    \right)
    \nonumber\\
    &=
    \mathcal{M}_{\mathrm{ISG}_1}
    \circ
    \mathcal{M}_{\mathcal{U}_1[\mathrm{ISG}_\infty]}
    \circ
    \mathcal{U}_1
    \circ
    \mathcal{U}_t.
    \label{eq:adjacent_measurement_layers}
\end{align}

We now reduce the two adjacent measurement layers. Since $\mathrm{ISG}_\infty$ is fixed by one Floquet period,
\begin{align}
    \mathcal{M}_{\mathrm{ISG}_1}
    \left[
        \mathcal{U}_1[\mathrm{ISG}_\infty]
    \right]
    =
    \mathcal{C}_F[\mathrm{ISG}_\infty]
    =
    \mathrm{ISG}_\infty.
    \label{eq:plateau_measurement_update}
\end{align}
Because $\mathcal{U}_1$ is unitary, the two groups have the same rank, so the measurement is rank preserving. The groups
$\mathcal{U}_1[\mathrm{ISG}_\infty]$
and
$\mathrm{ISG}_\infty$
therefore form a reversible pair in the sense of Ref.~\cite{aasen2023}.

Indeed, any element of $\mathcal{U}_1[\mathrm{ISG}_\infty]$ that commutes with $\mathrm{ISG}_\infty$ also commutes with $\mathrm{ISG}_1$ and is retained by the measurement in Eq.~\eqref{eq:plateau_measurement_update}; since the two groups have the same rank, the converse condition follows as well.

Let
\begin{align}
    \mathcal{H}
    \equiv
    \mathcal{U}_1[\mathrm{ISG}_\infty]
    \cap
    \mathrm{ISG}_\infty.
\end{align}
Applying Lemma~\ref{lem:common_stabilizer_subgroup} to Eq.~\eqref{eq:plateau_measurement_update} with
\begin{align}
    \mathcal{A}
    &=
    \mathcal{U}_1[\mathrm{ISG}_\infty],
    &
    \mathcal{B}
    &=
    \mathrm{ISG}_\infty,
    &
    \mathcal{R}
    &=
    \mathrm{ISG}_1,
\end{align}
gives
\begin{align}
    \mathcal{H}
    =
    \mathcal{U}_1[\mathrm{ISG}_\infty]
    \cap
    \mathrm{ISG}_1^\perp.
\end{align}
The measurement update therefore becomes
\begin{align}
    \mathrm{ISG}_\infty
    =
    \mathcal{H}
    +
    \mathrm{ISG}_1.
    \label{eq:plateau_group_decomposition}
\end{align}

As shown in Ref.~\cite{aasen2023}, the common subgroup $\mathcal{H}$ can be completed to generating sets of the reversible pair such that
\begin{align}
    \mathrm{ISG}_\infty
    &=
    \langle g_1,\ldots,g_\alpha\rangle
    +
    \mathcal{H},
    \nonumber\\
    \mathcal{U}_1[\mathrm{ISG}_\infty]
    &=
    \langle g'_1,\ldots,g'_\alpha\rangle
    +
    \mathcal{H},
    \label{eq:plateau_conjugate_generators}
\end{align}
where $g_i$ anticommutes with $g'_i$ and commutes with $g'_j$ for $i\neq j$.

Since Eq.~\eqref{eq:plateau_group_decomposition} states that $\mathrm{ISG}_1$, together with $\mathcal{H}$, generates all of $\mathrm{ISG}_\infty$, the complementary generators $g_1,\ldots,g_\alpha$ in Eq.~\eqref{eq:plateau_conjugate_generators} can be chosen from $\mathrm{ISG}_1$. Any remaining generators of $\mathrm{ISG}_1$ can then be chosen in $\mathcal{H}$, so
\begin{align}
    \mathrm{ISG}_1
    =
    \langle g_1,\ldots,g_\alpha\rangle
    +
    \left(
        \mathrm{ISG}_1\cap\mathcal{H}
    \right).
    \label{eq:ISG1_decomposition}
\end{align}

The output of
$\mathcal{M}_{\mathcal{U}_1[\mathrm{ISG}_\infty]}$
already contains
$\mathrm{ISG}_1\cap\mathcal{H}$,
so these measurements are redundant in the subsequent $\mathrm{ISG}_1$ layer. Applying the anticommuting case of Lemma~\ref{lem:single_measurement_reduction} sequentially to $g_1,\ldots,g_\alpha$ gives
\begin{align}
    \mathcal{M}_{\mathrm{ISG}_1}
    \circ
    \mathcal{M}_{\mathcal{U}_1[\mathrm{ISG}_\infty]}
    =
    \mathcal{M}_{\mathrm{ISG}_\infty}
    \circ
    \mathcal{U}',
    \label{eq:adjacent_measurement_reduction}
\end{align}
where
\begin{align}
    \mathcal{U}'
    \equiv
    \prod_{i=1}^{\alpha}
    \frac{g_i+g'_i}{\sqrt{2}}.
    \label{eq:reversible_measurement_unitary}
\end{align}
The factors commute because distinct conjugate pairs commute. By construction, each factor exchanges $g'_i$ and $g_i$ while leaving $\mathcal{H}$ unchanged. Therefore,
\begin{align}
    \mathcal{U}'
    \left[
        \mathcal{U}_1[\mathrm{ISG}_\infty]
    \right]
    =
    \mathrm{ISG}_\infty.
    \label{eq:Uprime_maps_plateau}
\end{align}

Substituting Eq.~\eqref{eq:adjacent_measurement_reduction} into Eq.~\eqref{eq:adjacent_measurement_layers} and defining
\begin{align}
    \mathcal{U}_F
    \equiv
    \mathcal{U}'
    \circ
    \mathcal{U}_1
    \label{eq:effective_floquet_unitary}
\end{align}
gives
\begin{align}
    \mathcal{C}_F^{t+1}
    =
    \mathcal{M}_{\mathrm{ISG}_\infty}
    \circ
    \mathcal{U}_F
    \circ
    \mathcal{U}_t.
\end{align}
Thus, we may choose
\begin{align}
    \mathcal{U}_{t+1}
    =
    \mathcal{U}_F
    \circ
    \mathcal{U}_t,
    \qquad
    t\geq t^*.
\end{align}
Moreover, Eq.~\eqref{eq:Uprime_maps_plateau} gives
\begin{align}
    \mathcal{U}_F[\mathrm{ISG}_\infty]
    &=
    \mathcal{U}'
    \left[
        \mathcal{U}_1[\mathrm{ISG}_\infty]
    \right]
    \nonumber\\
    &=
    \mathrm{ISG}_\infty,
\end{align}
showing that $\mathcal{U}_F$ preserves the plateau stabilizer group. Since the adjacent measurement layers are fixed, $\mathcal{U}'$ and $\mathcal{U}_F$ can be chosen independently of $t$.

\subsection{Action on Stabilizer Groups Containing the Plateau Group}

Let $\mathcal{S}_0$ be any stabilizer group satisfying
\begin{align}
    \mathrm{ISG}_\infty
    \subseteq
    \mathcal{S}_0.
\end{align}
Then
\begin{align}
    \mathcal{U}_1[\mathrm{ISG}_\infty]
    \subseteq
    \mathcal{U}_1[\mathcal{S}_0].
\end{align}
Measuring stabilizers already contained in a stabilizer group leaves that group unchanged, so
\begin{align}
    \mathcal{M}_{\mathcal{U}_1[\mathrm{ISG}_\infty]}
    \left[
        \mathcal{U}_1[\mathcal{S}_0]
    \right]
    =
    \mathcal{U}_1[\mathcal{S}_0].
\end{align}
We may therefore insert this measurement layer into the one period evolution and use Eq.~\eqref{eq:adjacent_measurement_reduction}:
\begin{align}
    \mathcal{C}_F[\mathcal{S}_0]
    &=
    \left(
        \mathcal{M}_{\mathrm{ISG}_1}
        \circ
        \mathcal{U}_1
    \right)[\mathcal{S}_0]
    \nonumber\\
    &=
    \left(
        \mathcal{M}_{\mathrm{ISG}_1}
        \circ
        \mathcal{M}_{\mathcal{U}_1[\mathrm{ISG}_\infty]}
        \circ
        \mathcal{U}_1
    \right)[\mathcal{S}_0]
    \nonumber\\
    &=
    \left(
        \mathcal{M}_{\mathrm{ISG}_\infty}
        \circ
        \mathcal{U}'
        \circ
        \mathcal{U}_1
    \right)[\mathcal{S}_0]
    \nonumber\\
    &=
    \left(
        \mathcal{M}_{\mathrm{ISG}_\infty}
        \circ
        \mathcal{U}_F
    \right)[\mathcal{S}_0].
\end{align}
Since $\mathcal{U}_F$ preserves the plateau stabilizer group,
\begin{align}
    \mathrm{ISG}_\infty
    =
    \mathcal{U}_F[\mathrm{ISG}_\infty]
    \subseteq
    \mathcal{U}_F[\mathcal{S}_0].
\end{align}
The final measurement
$\mathcal{M}_{\mathrm{ISG}_\infty}$
therefore leaves
$\mathcal{U}_F[\mathcal{S}_0]$
unchanged, giving
\begin{align}
    \mathcal{C}_F[\mathcal{S}_0]
    =
    \mathcal{U}_F[\mathcal{S}_0].
\end{align}
This establishes Eq.~\eqref{eq:CF_equals_UF}.

\section{Numerical Extraction of the Invertible Logical Map}
\label{app:numerical_map}

Once the stabilizer group reaches $\mathrm{ISG}_\infty$ during purification, one Floquet period induces an invertible map on the surviving logical Pauli operators. We numerically extract this logical map as follows. 
Let $\mathrm{ISG}_\infty$ have rank $N-k$, and choose an ordered logical basis
\begin{align}
    \bm{\mathcal L}
    =
    (\overline X_1,\ldots,\overline X_k,
    \overline Z_1,\ldots,\overline Z_k).
\end{align}
To tag these operators, we purify the plateau mixed state using ancillas $a_i$ so that the enlarged stabilizer group contains
\begin{align}
    \overline X_i X_{a_i},
    \qquad
    \overline Z_i Z_{a_i},
    \qquad
    i=1,\ldots,k.
\end{align}
We extend the Floquet circuit trivially to the ancillas, which remain unchanged throughout the evolution. Since the purified stabilizer group contains $\mathrm{ISG}_\infty$ on the physical qubits, Eq.~\eqref{eq:CF_equals_UF} applies to the enlarged system. One additional Floquet period acts only on the physical qubits, producing tagged operators
\begin{align}
    \overline X_i' X_{a_i},
    \qquad
    \overline Z_i' Z_{a_i}.
\end{align}
Gaussian elimination on the ancilla components isolates the $2k$ independent tags. Applying the same row operations to the physical components yields the evolved logical basis.

The $\alpha$th column of $M_F$ is the logical vector of the evolved basis operator $\mathcal L_\alpha'$. For $j=1,\ldots,k$, its components are obtained from the symplectic commutators
\begin{align}
    (M_F)_{j,\alpha}
    &=
    \omega(\mathcal L_\alpha',\overline Z_j),
    \\
    (M_F)_{k+j,\alpha}
    &=
    \omega(\mathcal L_\alpha',\overline X_j),
\end{align}
where $\omega(P,Q)=1$ for anticommuting Paulis and $0$ otherwise. With this convention, a logical Pauli operator represented by a column vector $\bm v$ evolves as
\begin{align}
    \bm v
    \longmapsto
    M_F\bm v.
\end{align}
Since Clifford evolution preserves Pauli commutation relations, $M_F$ is symplectic.

Extracting $M_F$ reduces the subsequent dynamics from the full $N$-physical-qubit stabilizer circuit to an invertible map acting on the $k$ surviving logical qubits, allowing us to access substantially larger system sizes. We use this logical representation to compute both the support radius and the Krylov dimension.

\section{Numerical Calculation of the Logical-Support Radius}
\label{app:logical_support_radius}

Appendix~\ref{app:numerical_map} constructs the effective Floquet map $M_F$ as an invertible $2k\times 2k$ matrix acting on the logical Pauli space. To calculate the logical-support radius, we find the smallest ball $B_R(r)$ that supports at least one physical Pauli representative of the evolved logical operator.

We choose physical Pauli representatives for the logical basis operators and a set of generators for the plateau stabilizer group. Starting from an SSLO $\hat P_r$, we evolve its binary vector in the logical basis using $M_F$. The corresponding product of the basis representatives gives a physical Pauli representative of the evolved operator. Multiplication by plateau stabilizers preserves its logical equivalence class.

We write the plateau stabilizer generators as rows of a binary matrix and order the physical qubits by decreasing distance from $r$. Each qubit contributes two consecutive columns for its $X$ and $Z$ components. The first columns therefore correspond to the most distant qubits, while the last two columns correspond to the qubit at $r$. The physical representative of the logical operator uses the same column ordering.

We perform Gaussian elimination on the stabilizer rows from left to right. We then use the resulting generators to reduce the logical representative in the same order, removing distant support whenever possible without restoring components already eliminated. The distance of the farthest remaining nonidentity Pauli factor gives the logical-support radius. No equivalent representative can be supported in a smaller ball.

For the distributions at late times in Figs.~\ref{fig:logical_support}(c) and \ref{fig:logical_support}(d), we perform this calculation after each complete Floquet period. For the trajectory in Fig.~\ref{fig:logical_support}(b), we also evaluate the radius after each of the four unitary layers within a period. At each such time, we use the physical representative and plateau stabilizer group at the corresponding stage of the circuit. This gives a time resolution of one quarter period.

\section{Additional Numerical Results}
\label{app:additional_numerics}

This appendix presents additional numerical results for the two-dimensional Floquet circuit studied in Sec.~\ref{sec:two_dimensional_transition}. We first examine the plateau entropy density obtained from purification dynamics, and then present additional results for the bipartite entanglement and antipodal mutual information.

\subsection{Purification Entropy Density}
\label{app:2d_purification}

We initialize the $L\times L$ torus in the maximally mixed state and evolve until the full-system entropy reaches its late-time plateau. Fig.~\ref{fig:2d_purification} shows the sample-averaged plateau entropy density. The curves for different system sizes nearly coincide and vary smoothly with $p$, with no visible feature near the entanglement transition at $p\simeq0.24$. The plateau entropy therefore remains extensive on both sides of the entanglement transition, consistent with the storage-region argument of Sec.~\ref{sec:one_dimension:purification}. 

\begin{figure}[t] 
    \centering \includegraphics[width=0.5\textwidth]{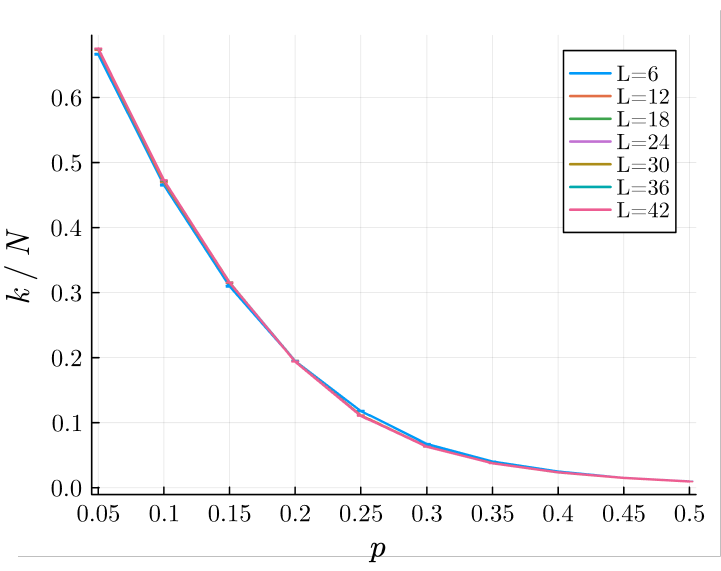} 
    \caption{Plateau entropy density in two dimensions. Sample-averaged late-time full-system entropy density for $L\times L$ systems with periodic boundary conditions initialized in the maximally mixed state. The curves show negligible system-size dependence and vary smoothly with $p$, including across the entanglement transition near $p\simeq0.24$, with no indication of a corresponding purification transition.}
\label{fig:2d_purification} 
\end{figure}

\subsection{Bipartite Entanglement Distributions}

As discussed in Sec.~\ref{sec:two_dimensional_transition}, the sample-averaged bipartite entropy clearly distinguishes the volume-law and area-law phases deep within each regime. Near the transition, however, it contains both long-range and short-range contributions that are difficult to separate.

To illustrate this issue, Fig.~\ref{fig:entanglement_distribution} compares two geometries with the same subsystem volume but different cut lengths. The first has $(L_x,L_y)=(32,32)$, while the second has $(L_x,L_y)=(64,16)$. In both cases, the subsystem on one side of the middle cut contains $16\times32$ qubits, but the length of the entanglement cut differs by a factor of two.

At low measurement rates, the distribution develops a peak at large entanglement values that reflects the long-range contribution present in the volume-law phase. At high measurement rates, most of the weight lies at smaller entanglement values associated with short-range entanglement near the cut. Near the transition, these contributions overlap strongly. This makes it difficult to extract the critical point from the bipartite entropy alone.

\begin{figure*}[t]
    \centering
    \includegraphics[width=0.85\textwidth]{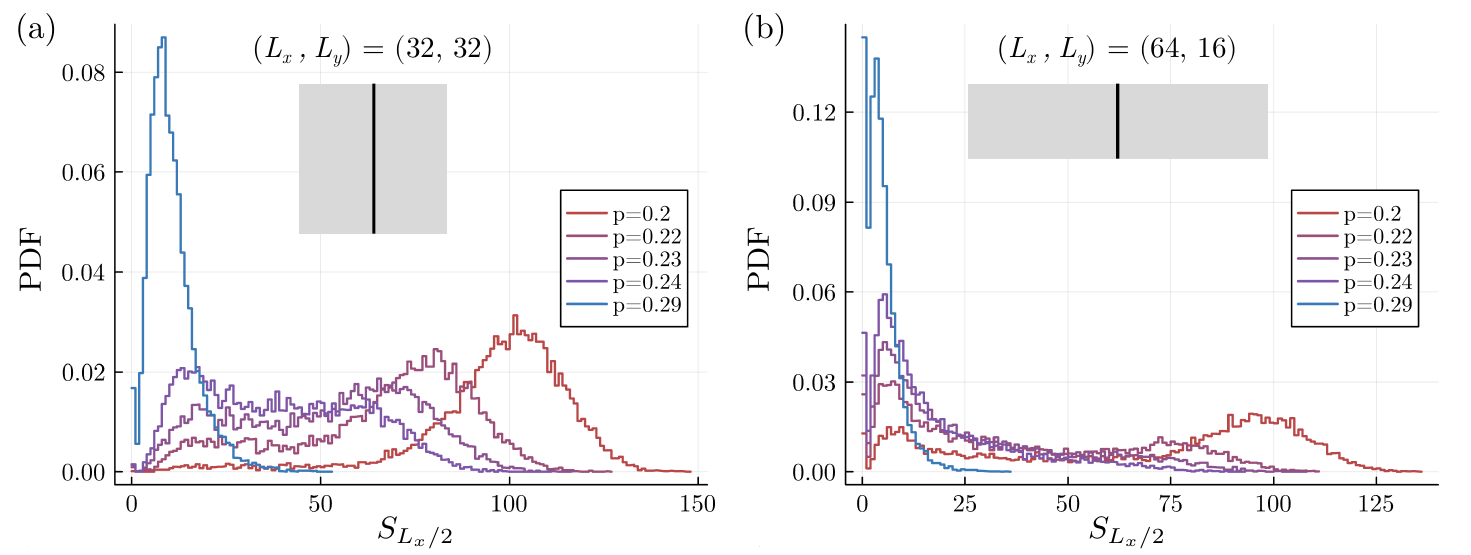}
    \caption{
    Distribution of the saturated bipartite entropy $S_{L_x/2}$ over circuit realizations for two geometries. The system sizes are (a) $(L_x,L_y)=(32,32)$ and (b) $(L_x,L_y)=(64,16)$. In both cases, the subsystem on one side of the middle cut contains $16\times32$ qubits, while the cut lengths are different. At low measurement rates, a peak at large entanglement reflects the long-range contribution in the volume-law phase. At high measurement rates, most of the weight is concentrated at smaller values associated with short-range entanglement. Near the transition, the two contributions overlap.
    }
    \label{fig:entanglement_distribution}
\end{figure*}

\subsection{Entanglement Growth in Individual Realizations}
\label{app:entanglement_growth}

Fig.~\ref{fig:entanglement_growth} shows the time evolution of the bipartite entanglement entropy $S_{L_x/2}(t)$ for representative circuit realizations at $L=64$ and $p=0.21$, $0.24$, and $0.29$. The entanglement grows toward its saturation value with only small $O(1)$ fluctuations. We do not observe any substantial decrease during the growth, apart from these small fluctuations. We evolve each realization for a sufficiently long time that the entanglement is clearly saturated. We then estimate the saturated entropy by averaging over the final time window. The saturation time $t_{\mathrm{sat}}$ is defined as the first time at which the entanglement reaches within one unit below this estimate. This choice makes $t_{\mathrm{sat}}$ less sensitive to small fluctuations near saturation. The median of $t_{\mathrm{sat}}$ over circuit realizations is used to compute the effective dynamical exponent $z_{\mathrm{eff}}$ shown in the inset of Fig.~\ref{fig:2d_bipartite}.

\begin{figure*}[t]
    \centering
    \includegraphics[width=0.9\textwidth]{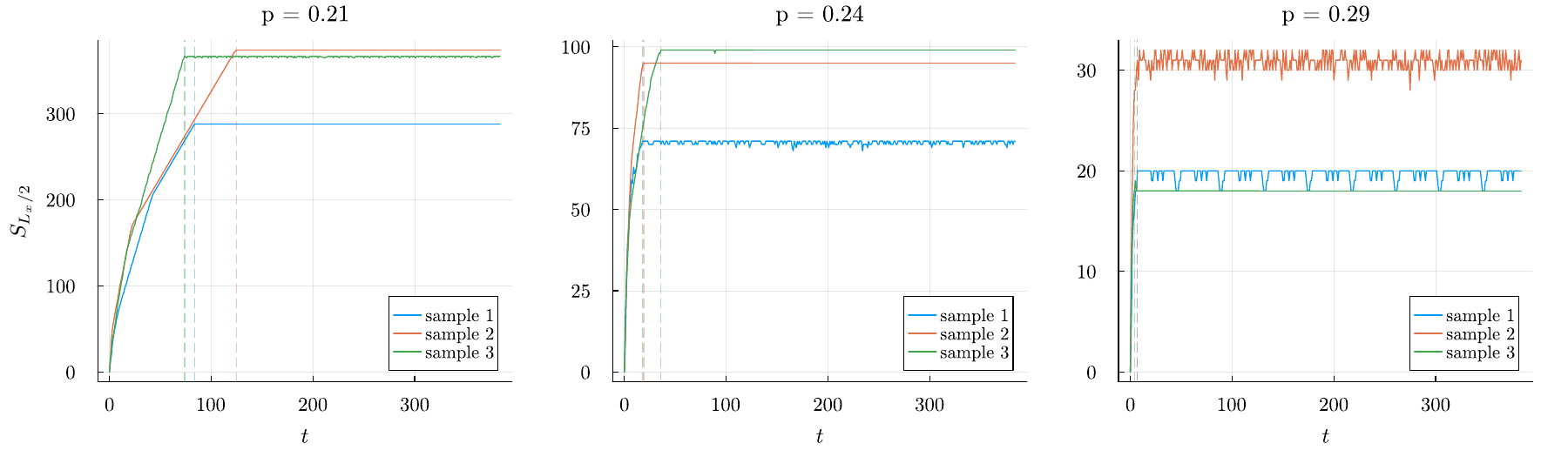}
    \caption{Entanglement growth in individual circuit realizations at $L=64$. Three representative realizations are shown for each of $p=0.21$, $0.24$, and $0.29$. Vertical dashed lines indicate the corresponding saturation times $t_{\mathrm{sat}}$.}
    \label{fig:entanglement_growth}
\end{figure*}

\subsection{Sample-to-Sample Fluctuations of the AMI}

The AMI exhibits strong sample-to-sample fluctuations near the transition. Fig.~\ref{fig:AMI_zero} shows two measures of the weight at small AMI. The left panel shows the fraction of circuit realizations whose AMI remains zero throughout the late-time averaging window. The right panel shows the fraction whose time-averaged saturated AMI is smaller than one. Although the instantaneous AMI is integer valued, its late-time average need not be. For example, an AMI that oscillates between zero and one has a time average between these values.

\begin{figure}[t]
    \centering
    \includegraphics[width=0.4\textwidth]{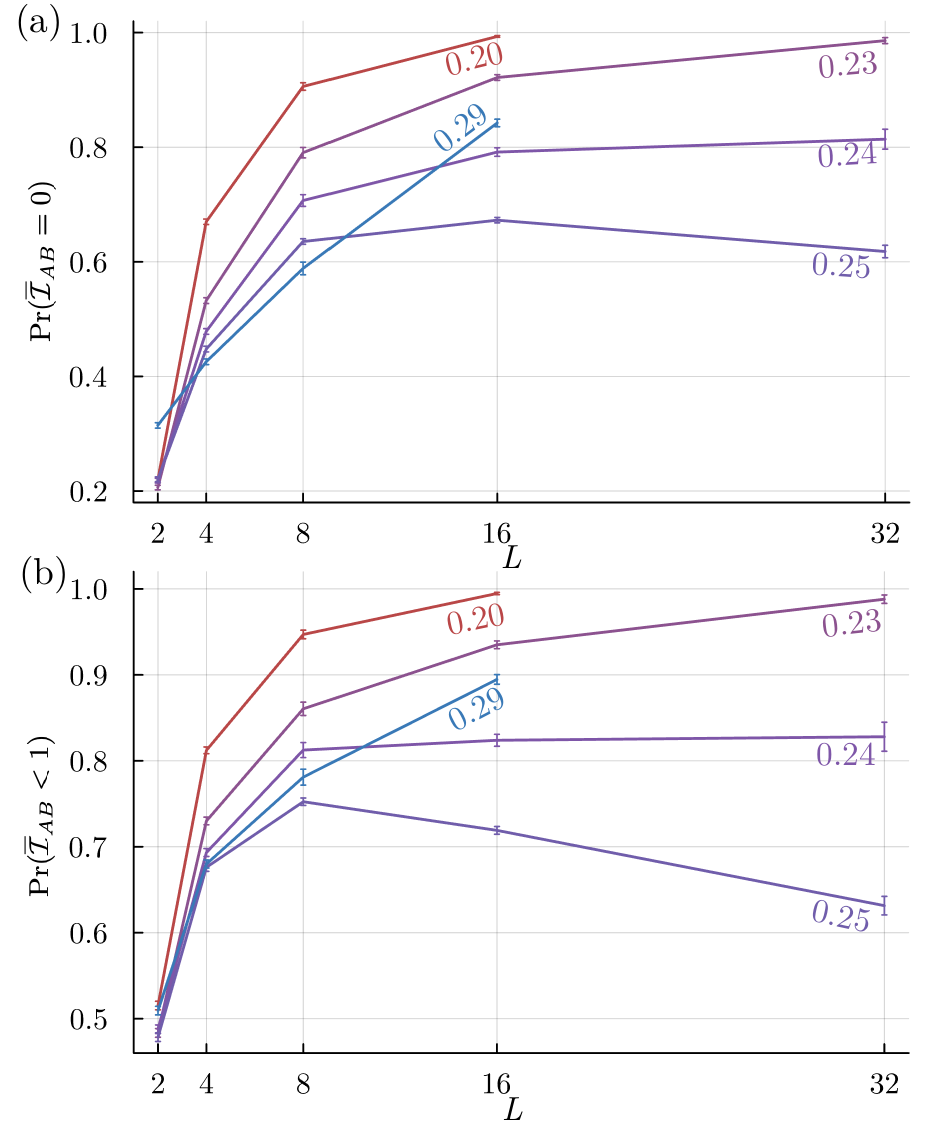}
    \caption{(a) Fraction of circuit realizations with $\overline{\mathcal I}_{AB}=0$. (b) Fraction with $\overline{\mathcal I}_{AB}<1$, where $\overline{\mathcal I}_{AB}$ denotes the average over the late-time window. Non-integer values of $\overline{\mathcal I}_{AB}$ can arise from temporal fluctuations between integer-valued AMI. Error bars show standard errors.}
    \label{fig:AMI_zero}
\end{figure}

The full distribution is shown through its complementary cumulative distribution function in Fig.~\ref{fig:AMI_cdf}. Near the putative critical point, the distribution changes strongly with system size and retains substantial weight at both zero and large AMI. These broad sample-to-sample fluctuations limit the precision with which the critical point can be determined from the sample average alone.

\begin{figure}[t]
    \centering
    \includegraphics[width=0.5\textwidth]{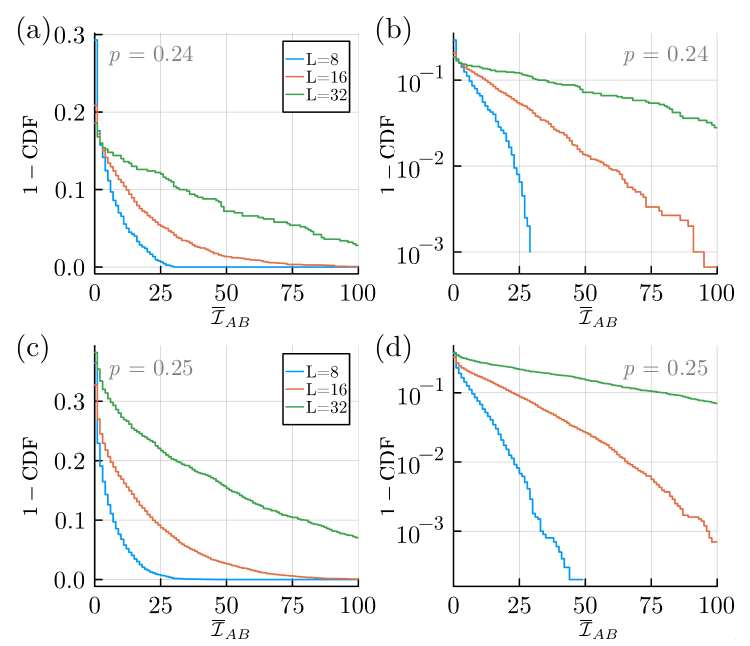}
    \caption{Complementary Cumulative distribution function $1-\mathrm{CDF}$ of the late-time-averaged AMI for several system sizes, near the critical value of $p$. The top row shows $p=0.24$ and the bottom row shows $p=0.25$. The left panels use a linear scale and the right panels use a logarithmic scale for $1-\mathrm{CDF}$.
    }
    \label{fig:AMI_cdf}
\end{figure}

\subsection{AMI in Spacetime-Random Circuits}
\label{app:spacetime_random_AMI}

For comparison, we compute the same antipodal mutual information in the spacetime-random $(2+1)$-dimensional Clifford circuit of Ref.~\cite{lunt2021}. We use the same four nearest-neighbor gate directions, $\{\hat x,\hat y,-\hat x,-\hat y\}$, as in our Floquet model, and every site is independently measured in the $\hat Z$ basis with probability $p$ after each unitary layer. The difference is that this circuit is not Floquet: the Clifford gates and measurement locations are resampled as the dynamics proceeds rather than fixed and repeated in time. We use the same AMI geometry and late-time averaging procedure as in Fig.~\ref{fig:AMI}.

Fig.~\ref{fig:AMI_random} shows the resulting late-time AMI. A peak develops near the transition, but its height remains below one and shows no strong growth over the accessible system sizes. This $O(1)$ critical AMI contrasts with the much stronger size dependence found in the Floquet circuit.

\begin{figure}[t]
    \centering
    \includegraphics[width=0.5\textwidth]{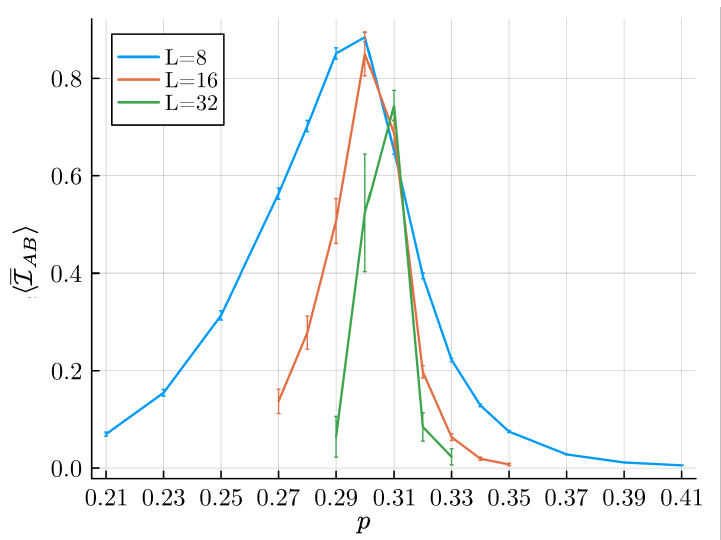}
    \caption{Late-time antipodal mutual information in the spacetime-random monitored Clifford circuit~\cite{lunt2021}. The peak remains below one for the accessible sizes $L=8$, $16$, and $32$ and shows no growth with system size.} 
    \label{fig:AMI_random}
\end{figure}

\subsection{Mean Krylov Dimension}
\label{app:krylov_mean}

In the main text, we introduced $\widetilde K=\langle K^2\rangle/\langle K\rangle$ to separate the characteristic IR scale from the overall IR weight $W$ and the finite UV contribution. This requires that $W$ does not decrease too rapidly with system size at the critical point. Fig.~\ref{fig:krylov_mean} shows $\log_2\langle K\rangle$ versus $\log_2 L$. Near the critical point, the slope remains clearly nonzero over the accessible system sizes, showing that the growing IR contribution remains significant. For each circuit realization, $K$ is first averaged uniformly over its SSLOs, and these sample means are then averaged uniformly over circuit realizations. The ratio $\widetilde K$ therefore provides a cleaner measure of the IR scale with reduced finite-size effects.

\begin{figure}[t]
    \centering
    \includegraphics[width=0.5\textwidth]{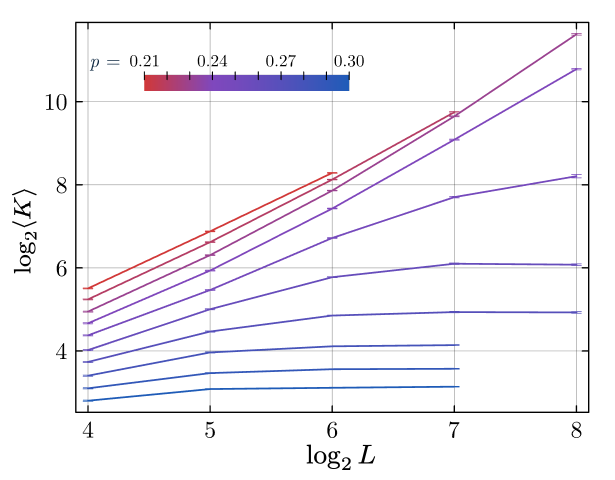}
    \caption{Mean Krylov dimension $\langle K\rangle$ as a function of system size for different measurement rates $p$, shown on a log-log scale. In the delocalized phase, the slope remains nonzero, with visible finite-size effects as it evolves with increasing system size. Near the transition, the slope remains nonzero over the accessible sizes, while in the localized phase it approaches zero.}
    \label{fig:krylov_mean}
\end{figure}

\FloatBarrier
\bibliography{main}

\end{document}